\documentclass{article}

\usepackage[preprint]{neurips_2026}

\usepackage[utf8]{inputenc} 
\usepackage[T1]{fontenc}    
\usepackage{hyperref}       
\usepackage{url}            
\usepackage{booktabs}       
\usepackage{amsfonts}       
\usepackage{nicefrac}       
\usepackage{microtype}      
\usepackage{xcolor}         

\usepackage{amsmath}
\usepackage{amssymb}
\usepackage{mathtools}
\usepackage{amsthm}

\theoremstyle{plain}
\newtheorem{theorem}{Theorem}[section]

\theoremstyle{definition}
\newtheorem{definition}[theorem]{Definition}

\theoremstyle{remark}
\newtheorem{remark}[theorem]{Remark}

\usepackage{microtype}
\usepackage{graphicx}
\usepackage{subcaption}
\usepackage{booktabs} 
\usepackage{enumitem}

\title{High-Dimensional Latents Should Be Diagnosed Through Phase Structure}

\author{1\textsuperscript{st} Alejandro Asc\'arate\thanks{
\small School of Electrical Engineering and Robotics, Faculty of Engineering,\\
\small $-\;\,\,\,\text{Queensland}$ University of Technology, Brisbane, Queensland, Australia} \\
a.ascaratecastro@hdr.qut.edu.au
\and
2\textsuperscript{nd}  L\'eo Lebrat$^{*}$ \\
leo.lebrat@qut.edu.au
\and
3\textsuperscript{rd} Rodrigo Santa Cruz$^{*}$ \\
rodrigo.santacruz@qut.edu.au
\and
4\textsuperscript{th} Clinton Fookes$^{*}$ \\
c.fookes@qut.edu.au
\and
5\textsuperscript{th} Olivier Salvado$^{*}$ \\
olivier.salvado@qut.edu.au
}

\begin{document}

\maketitle

\begin{abstract}
We study autoencoder and variational-autoencoder latent spaces through the lens of spin-glass theory. The paper has two components. First, we formalize a latent-space spin-glass dictionary: for a fixed decoder, the reconstruction term together with a hyperspherical coordinates prior induces a Hamiltonian on the latent sphere, where latent coordinates play the role of continuous spins and the prior acts as an external magnetic field. This allows us to import operational spin-glass diagnostics---overlap distributions, susceptibility, and block-spin coarse-graining---to detect ordered, disordered, and edge-of-stability phases in trained latent representations. Second, we show that deliberately driving the latent system toward the edge-of-stability of the topological trivialization regime has concrete downstream consequences. In generation, hyperspherical compression improves the reconstruction--generation trade-off on CIFAR-10 and CelebA64, yielding lower self-FID while preserving or improving reconstruction. In anomaly detection, the same semi-ordered latent geometry improves both fully unsupervised and conditional OOD detection, including real-world Mars Rover and Galaxy Zoo datasets, as well as CIFAR-10/100 and Imagenette-based OOD benchmarks. We therefore advocate a phase-aware evaluation paradigm for AEs/VAEs, in which spin-glass observables complement standard ML metrics and expose the latent regimes that underlie downstream success or failure in many cases.
\end{abstract}

\section{Introduction}

\begin{figure}[!h]
    \centering
    \includegraphics[width=1.0\linewidth]{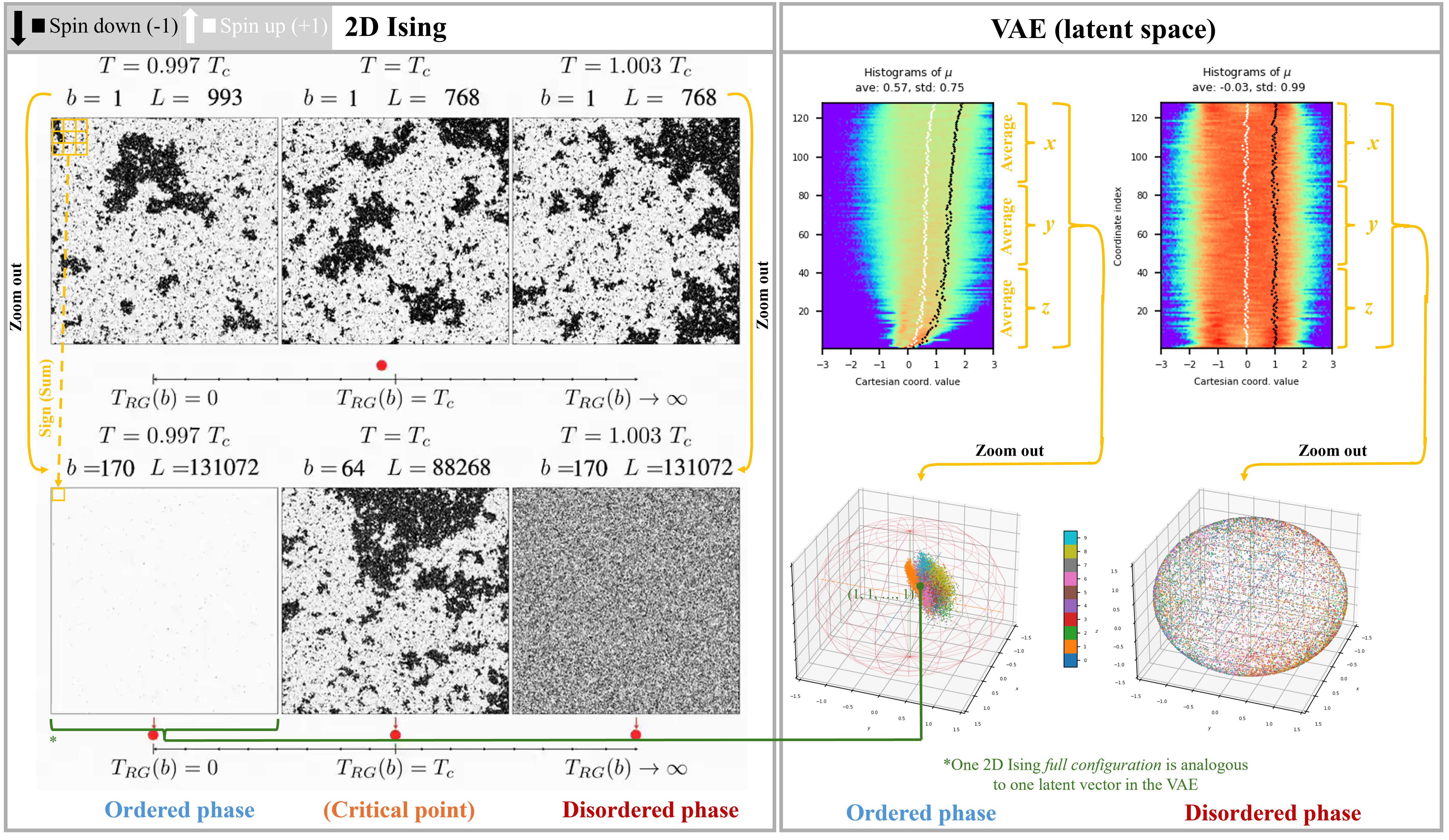}
    \caption{\textbf{Detecting phases via Block--spin--like coarse--graining of latent coordinates $\mu$} (neighboring--dimension averaging down to $\mathbb{R}^3$ from an initial $128-$dimensional latent). \textbf{Left:} a standard 2D Ising model Block--spin coarse--graining (dashed yellow line illustrating a $3\times3\to1$ reduction, $\mathrm{new\,spin}=sign(\Sigma^9 \,\mathrm{spins})$, i.e., $b=3$; these are not the real scale factors $b$ and length scale $L$ of the lattice in the displayed images, the actual values appear on top of each) showing its three distinct phases depending on the temperature $T$ of the initial configuration; this `Zooming out' process (since one loses \textit{local} details) makes the system look like one in a different, effective temperature $T_{RG}(b)$ (red dots). \textbf{Right:} Our version of Block--spin--like coarse--graining in latent space, which cleanly visualizes disorder$\to$order. Results from one MNIST run for illustrative purposes, with the one at the left in the full nested RSB by hyperspherical coordinates mode, while at the right we have the standard VAE with Gaussian prior $\mathcal{N}(0,I)$; a simple $k-$NN on the full latent reveals a classification accuracy of $>90\%$ in \textit{both} phases, but it is \textit{only} in the ordered one where the Block–-spin--like coarse–graining transformation reveals this clustering due to the stability of the samples in the semi-ordered phase under this type of transformations. See \ref{sec:spin-glass-glossary} for terminology.}
    \label{renorm_diags}
\end{figure}


Spin--glass theory offers a language for how collective behavior, phases, and rugged landscapes emerge from many weak, high--order interactions. Its mean--field formulations and the associated tools---overlaps, susceptibilities, landscape complexity, and renormalization---have been developed in physics \cite{MPV1987,CugliandoloKurchan1993,FranzParisi1995,FyodorovNadal2012,Fyodorov2016} and placed on rigorous footing in mathematics \cite{Talagrand2011,Panchenko2013,AuffingerBenArousCerny2013,AB13AOP,Subag2017}. In machine learning, classic connections include Hopfield networks and Boltzmann machines, and more recent work has compared deep--network loss surfaces and training dynamics to glassy systems \cite{Hopfield1982,AckleyHintonSejnowski1985,AmitGutfreundSompolinsky1985,Choromanska2015,BaityJesiICML2018,BaityJesiJSTAT2019,Dauphin2014}. Most analyses, however, either live in \emph{weight space} or use spin--glass ideas as qualitative metaphors.

\paragraph{Latent--centric viewpoint.}
We take a different route: we treat the \emph{latent space} of an autoencoder/VAE as the spin configuration space. For a fixed decoder snapshot and input $x$, the reconstruction term together with a simple prior on the latent mean $\mu$ induces a scalar \emph{latent energy} $H_x(\mu)$ on a high--dimensional hypersphere. In this dictionary, latent coordinates play the role of continuous spins; the prior acts as an external field; and the encoder’s evolution under training corresponds to a noisy gradient flow (Langevin--like) on the induced energy landscape. This framing allows us to port \emph{concrete} spin--glass observables to latent space and ask falsifiable, data--driven questions about \emph{phases} of $\underline{\text{\emph{representation}}}$ and to investigate their impact on standard ML/AI tasks (in this work, we mainly focus on two, Generation and Anomaly Detection, both in the context of VAEs).

\paragraph{What we show.}
The paper is divided in two main sections:
\begin{itemize}
    \item In the first one \ref{sec:encoder_spin_glass}-\ref{sec:spin_glass_tools}, we formalize the VAE-latent$\iff$spin–glass mapping, and then we import to \emph{latent--space} two key tools from spin--glass theory which are used to detect and see phase transitions in the latter (Block–spin coarse–graining and the overlap order parameter). Armed with these tools, we diagnose in latent space the presence of ordered, disordered, and edge-of-stability phases, and whose transitions we control thanks to a specific angular-based (hyperspherical coordinates) KLD-like term, where the prior is akin to the application of an external magnetic field in spin glasses. Crucially, we use this driver to make the system go to an edge-of-stability phase.
    
    \item In the second part \ref{results}, we show how this driving of the latent representations to this phase causes direct improvements on two key ML/AI tasks (Generation and Anomaly Detection).
\end{itemize} 

\paragraph{Position.}
We argue that high-dimensional latent spaces in AEs/VAEs should be treated as finite-size statistical-mechanical systems rather than merely as Euclidean embeddings. In particular, latent representations should be diagnosed by phase-sensitive observables---overlap distributions, susceptibilities, coarse-graining stability, and hyperspherical angular order parameters---before being evaluated only through downstream task metrics. Our central claim is that many failures of direct VAE generation and latent-space anomaly detection are symptoms of a disordered high-volume latent phase, while useful representations arise near an intermediate, partially ordered edge-of-stability regime.

 \paragraph{Glossary of spin glass terminology} For the benefit of the non-expert reader, we offer in the appendix an extensive glossary of spin glass terms and concepts used through the paper \ref{sec:spin-glass-glossary}. We strongly encourage its browsing.

\section{Related Work}

\paragraph{Statistical physics of spin glasses.}
The foundations of mean-field spin-glass theory and replica symmetry breaking were developed in the physics literature, culminating in the monograph of Mézard, Parisi, and Virasoro \citep{MPV1987}; see \cite{MontanariSen2024Friendly} for an introductory account. Nonequilibrium glassy dynamics, aging, and fluctuation–dissipation violations were analyzed by Cugliandolo and Kurchan \citep{CugliandoloKurchan1993}, while the Franz--Parisi potential provided a geometric, overlap-based effective potential linking metastability and landscape structure \citep{FranzParisi1995}. More recently, random-landscape methods have clarified topology changes (``topology trivialization'') and the counting of stationary points via Kac–Rice and random-matrix tools \citep{FyodorovNadal2012,Fyodorov2016}. See \cite{RosBenArousBiroliCammarota2019PRX} for an insightful visualization of the emergence of topological trivialization as the strength of a bias (e.g., an external magnetic field) increases.

\paragraph{Mathematical theory.}
On the rigorous side, Talagrand’s volumes \citep{Talagrand2011} and Panchenko’s monograph \citep{Panchenko2013} established the Parisi formula and the structure of the SK model. For spherical $p$-spins, the complexity of critical points and geometry of high-dimensional random fields were characterized in a series of works using Kac–Rice and GOE techniques \citep{AuffingerBenArousCerny2013,AB13AOP}. Subsequent developments describe the organization of low-temperature Gibbs measures and energy landscapes \citep{Subag2017}.

\paragraph{Historical ML/AI connections.}
Spin-glass ideas have influenced neural networks since the early days: Hopfield’s associative memory \citep{Hopfield1982}, Boltzmann machines \citep{AckleyHintonSejnowski1985}, and the Amit, Gutfreund, Sompolinsky program relating storage capacity and phases \citep{AmitGutfreundSompolinsky1985} are canonical references linking disordered statistical mechanics and learning.

\paragraph{Modern ML/AI: loss landscapes and glassy analogies.}
A line of work connects deep-network loss surfaces to spherical spin-glass Hamiltonians and analyzes their critical points \citep{Choromanska2015}. Beyond static analogies, Baity-Jesi, Ben~Arous, LeCun, Wyart, Biroli and collaborators compared training dynamics in deep nets with glassy systems, reporting flat directions and under/over-parameterized ``phases'' \citep{BaityJesiICML2018,BaityJesiJSTAT2019}. Related studies probe saddle geometry and Hessian spectra in deep models \citep{Dauphin2014}.

\paragraph{Relation of this work to spin-glass analyses of neural networks.}

Our approach is different in two key respects. First, we do not
model the \emph{weights} as spins and the training loss as a
weight-space Hamiltonian. Instead, we treat the \emph{latent
representations} of an autoencoder as spins and show that, for a
fixed decoder snapshot, the latent energy lies in the mixed spherical
$p$--spin universality class. Second, rather than using spin-glass
theory purely as an analogy or a source of qualitative metaphors, we
use it as a source of \emph{concrete predictions and observable diagnostics}
for the latent phase, and we verify these predictions empirically on
trained models.


\section{Autoencoder networks as latent spin-glass systems}
\label{sec:encoder_spin_glass}

We begin by recalling the basic objects of spin-glass theory and
then explain how we interpret a neural encoder network as a
high-dimensional spin system with a disordered Hamiltonian. In our
view, the latent coordinates play the role of spins, the neural
weights encode quenched disorder, and the training dynamics of the
encoder--decoder pair is naturally compared to Langevin dynamics in
an energy landscape. This section sets up the dictionary we will use
throughout the paper; \ref{sec:pspin_and_main_theorem} turns this into a precise
theorem via a Taylor expansion of the decoder.

\subsection{Classical spin systems and Gibbs measures}

A classical spin system consists of a large number
$N \gg 1$ of spins $s_i$ placed on a lattice or a graph. In the
simplest Ising case $s_i \in \{-1,+1\}$, while in the spherical
model one considers continuous spins
$s = (s_1,\ldots,s_N) \in \mathbb{R}^N$ constrained to the unit
sphere, $S^{N-1}$.
The state of the system is a configuration $s \in S^{N-1}$, and its
energy is given by a Hamiltonian $H(s)$; in spin glasses $H$ is a
random function of $s$ with quenched disorder coming from random
couplings.

Given $H$, the equilibrium statistics at inverse temperature
$\beta = 1/T$ are described by the Gibbs measure
\begin{equation}
  \label{eq:gibbs}
  \pi_\beta(s)
  \;=\;
  \frac{1}{Z(\beta)} \exp\bigl[-\beta H(s)\bigr],
\end{equation}
\begin{equation}\label{eq:part}
  Z(\beta) = \int_{S^{N-1}} \exp\bigl[-\beta H(s)\bigr] \, d\mu(s),
\end{equation}
where $\mu$ is the uniform measure on the sphere. At high
temperature $\beta \to 0^+$ the Gibbs measure approaches the uniform
distribution on $S^{N-1}$ (fully disordered phase). At low
temperature large portions of configuration space are suppressed and
typical spins concentrate on low-energy regions (ordered or glassy
phases).

A key dynamical fact is that the Gibbs measure is stationary for
Langevin dynamics, i.e.\ for the stochastic differential equation
on $S^{N-1}$:
\begin{equation}
  \label{eq:langevin}
  d s_t
  \;=\;
  - \nabla H(s_t)\, dt
  + \sqrt{2/\beta}\, dB_t,
\end{equation}
where $B_t$ is Brownian motion restricted to the sphere. Informally,
the gradient term performs steepest descent on the energy landscape,
while the noise term injects thermal fluctuations of intensity
$1/\beta$; together they drive the system toward the Gibbs
distribution~\eqref{eq:gibbs}.

\subsection{Neural encoders as energy-based models on the latent space}

We now describe how an encoder network can be viewed as defining a
spin-glass-like Hamiltonian on its latent space. Let
$\mathcal{X} \subset \mathbb{R}^{d_x}$ be the input space, and let
$\mathcal{Z} \subset \mathbb{R}^{d_z}$ be the latent space, with
$d_z$ large. A parametrized encoder--decoder pair consists of
\[
  E_\theta : \mathcal{X} \to \mathcal{Z},
  \qquad
  D_\theta : \mathcal{Z} \to \mathcal{X},
\]
where $\theta \in \mathbb{R}^P$ collects all learnable weights and
biases. For concreteness one may think of a VAE encoder outputting
a mean $\mu_\theta(x)$ and (optionally) a covariance, but here we
focus on the deterministic mean code
\[
  z_\theta(x) \;=\; E_\theta(x) \in \mathcal{Z}.
\]

We interpret each coordinate $z_i$ as a continuous spin, and the
latent code $z\in\mathcal{Z}$ as a spin configuration of dimension
$N=d_z$.

Given a reconstruction loss $\ell(x, \hat x)$ and a latent
regularization term $R(z)$ (which, in VAEs, comes from the KLD term
to the prior), we define for each fixed input $x$ an \emph{effective
latent energy}
\begin{equation}
  \label{eq:latent_energy}
  H_x(z; \theta)
  \;=\;
  \ell\bigl(x, D_\theta(z)\bigr)
  \;+\;
  R(z).
\end{equation}
For a given decoder snapshot $(D_\theta)$, this is a scalar function
on $\mathcal{Z}$; when restricted to a sphere
$S^{d_z-1} \subset \mathcal{Z}$ (e.g.\ after normalizing by
$\|z\|_2$) it plays exactly the role of a Hamiltonian $H(s)$ in
\eqref{eq:gibbs}, with the latent coordinates as spins.

In this view the encoder is an \emph{energy-based model}: for a fixed
parameter vector $\theta$, the latent space is equipped with an
energy function $H_x(\cdot;\theta)$, and the encoder map
$E_\theta(\cdot)$ learns to send each input $x$ to low-energy
configurations $z_\theta(x)$. The disorder in the Hamiltonian is
provided by the random initialization of $\theta$ and by the
complex, data-dependent evolution of $\theta$ during training. More details in \ref{sec:pspin_and_main_theorem}.

\subsection{KLD-like latent regularization term via hyperspherical coordinates}

We now introduce hyperspherical coordinates in the KLD formulation. We start with the Cartesian coordinates \((\mu_{i},\sigma_{i})\), given by the encoder, and transform these to their hyperspherical counterparts \((\overset{\mu}{r},\overset{\mu}{\varphi_{k}};\overset{\sigma}{r},\overset{\sigma}{\varphi_{k}})\) with $r$ a scalar and $k$ the index of the $n-1$ spherical angles. The KLD-like objective becomes for the angles $\varphi_k$,

\begin{align}
\text{KLD}_{\text{HSphCoords}}^{w/Prior}(\varphi_k) = \nonumber
\sum_{k=1}^{n-1} &\bigg( \alpha_{\sigma,k} \left(\mathbb{E}_{b}[\mathrm{cos}\,\overset{\sigma}{\varphi_{k}}] - a_{\sigma,k} \right)^{2} \nonumber+ \beta_{\sigma,k}\left(\sigma_{b}[\mathrm{cos}\,\overset{\sigma}{\varphi_{k}}]-b_{\sigma,k}\right)^{2} \nonumber\\
&+ \alpha_{\mu,k}\!\!\left(\mathbb{E}_{b}[\mathrm{cos}\,\overset{\mu}{\varphi_{k}}]-a_{\mu,k}\right)^{2}\!\!\!\! + \beta_{\mu,k}\left(\sigma_{b}[\mathrm{cos}\,\overset{\mu}{\varphi_{k}}]-b_{\mu,k}\right)^{2} \!\bigg),\!\!\! 
\end{align}
with the priors for the mean over the batch $a_{i,j}$, the standard deviation over the batch $b_{i,j}$, and the gains for each term $\alpha_{i,j},\,\beta_{i,j}$, for $i\in\{\sigma,\mu\}$ and $j\in\{1,...,n-1,r\}$. The prior values $a_{i,j}$ act as the intensity of applied external magnetic fields (see \ref{mfield}) but now on the angular directions. This will be key due to the first bullet point in section \ref{over} below.

\section{Spin-glass tools in latent-space and latent-phase diagnosis}
\label{sec:spin_glass_tools}


In this section we
recall several standard tools from spin-glass theory and explain how
we will apply them to our latent setting. Each tool gives rise to
concrete, testable predictions about the behaviour of trained
autoencoders as we vary a single control parameter (e.g.\ the
compression strength or an effective inverse temperature).

\subsection{Overlaps, order parameters, and phase structure}\label{over}

In spherical $p$--spin models, the basic two-replica observable is
the \emph{overlap} between two spin configurations $s,s'\in S^{N-1}$:
\begin{equation}
  q(s,s')
  \;=\;
  \frac{1}{N} \sum_{i=1}^N s_i s'_i
  \;=\;
  \frac{1}{N} \langle s, s' \rangle.
\end{equation}
If we draw two configurations independently (called `replicas') from the Gibbs measure
$\pi_{\beta,N}$ in~\eqref{eq:gibbs}, their overlap is a random
variable $q$ whose law $P_\beta(q)$ encodes the phase structure:
\begin{itemize}
  \item In a high-temperature paramagnetic phase (\textit{disorder}) one typically has
        $P_\beta(q)$ sharply concentrated near $q=0$ (no global
        alignment between replicas). This is equivalent to the \emph{angle} between replicas being close to $\boxed{\pi/2}$ (often called `almost orthogonality', and due to the \textit{high dimensionality} of the hypersphere and the nature of the uniform distribution on it in that case; cf. eqs.\ref{eq:gibbs},\ref{eq:part}), or their \emph{separation in Euclidean distance} being close to $\sqrt{2}$ (or $\sqrt{2}\times\sqrt{N}$, if the hypersphere is not normalized).
  \item In a low-temperature ferromagnetic phase (\textit{order}) one has
        a peak at $q\approx m^2>0$, where $m$ is the spontaneous
        magnetisation.
  \item In a \textit{glassy} phase with replica-symmetry breaking (RSB) the
        overlap distribution $P_\beta(q)$ becomes nontrivial (e.g.\
        \emph{multiple peaks} or a \emph{continuous support on an interval}),
        reflecting a hierarchy of pure states.
\end{itemize}
The overlap is therefore an order parameter: changes in the shape of
$P_\beta(q)$ signal phase transitions as the control parameter (here,
$\beta$) is varied.

In our latent setting we fix a dataset and a trained encoder $E_\theta$
and consider the normalised latent spins $ s(x)\in S^{N-1},$
for inputs $x$ drawn from some distribution (training or validation).
We define the \emph{latent overlap} between two datapoints $x,x'$ as
\begin{equation}
  \label{eq:latent_overlap}
  R(x,x')
  \;=\;
  \frac{1}{N} \bigl\langle s(x), s(x') \bigr\rangle,
\end{equation}
and denote by $P(R)$ the empirical distribution of $R(x,x')$ over
random pairs of datapoints.


In our setting with hyperspherical coordinates, we can set a prior for the $\varphi_k$ that \textit{forces the latent replica-samples away from the almost orthogonality of the disordered phase}, by taking $\boxed{a_{\mu,k}\neq0,\,\forall k}$. We call this hypervolume compression via magnetic fields in hyperspherical coordinates. Empirically, we track $P(R)$ and summary statistics such as
$\mathbb{E}[R]$ and $\mathrm{Var}(R)$ across training and across
different compression strengths, and identify phase transitions via
qualitative changes in these overlap-based order parameters. See Fig.\ref{rep_angle_hists}.

\begin{figure}[!h]
    \centering
    \includegraphics[width=0.6\linewidth]{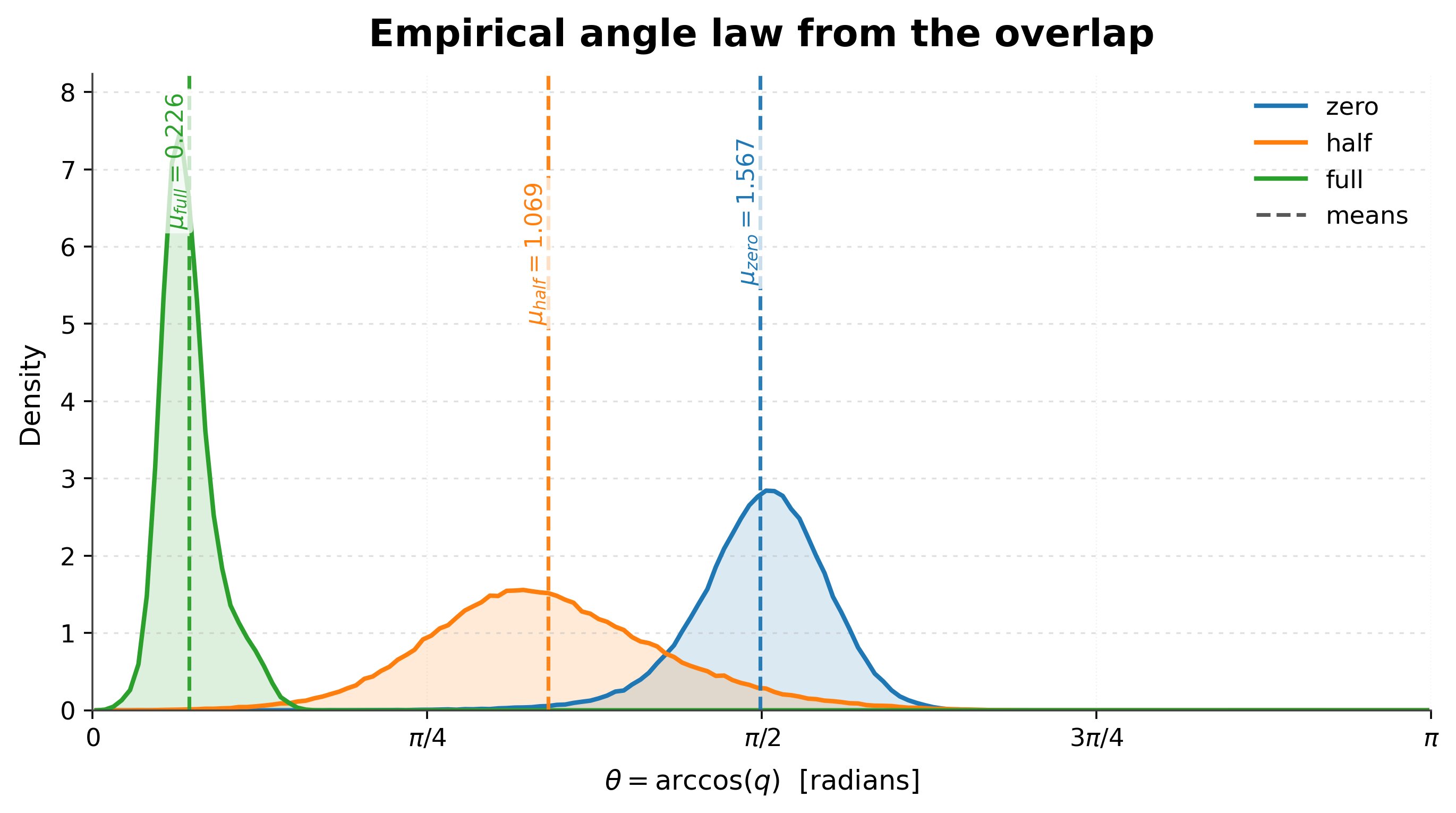}
    \caption{\textbf{Replica angle}. From right to left, the direction of the applied magnetic field in each experiment is shifted from the equator (`zero'--compression mode), to $(1,1,\dots,1)$ (`half'), and then towards the north pole (`full') to bias the transition during training (first two, right to left, are the experiments of Fig.\ref{renorm_diags}).}
    \label{rep_angle_hists}
\end{figure}

\subsection{Renormalisation-group and block-spin transformations}

Renormalisation-group (RG) techniques study how the statistics of a
spin system change under coarse-graining. A standard construction
for Ising or spherical models is the \emph{block-spin} transform \cite{Wilson1979ManyScales,Yeomans1992}:
partition the spins into blocks of size $b$, and define a new
coarse-grained spin variable in each block by averaging the original
spins. Fixed points of this transformation correspond to phases, and
critical points are characterised by scale invariance and unstable
directions of the RG flow.

In our latent setting, we can implement a simple block-spin
transformation by grouping neighbouring latent coordinates into a
small number of blocks and averaging them. Concretely, for
$N=d_z$ latent coordinates and a partition of $\{1,\dots,N\}$ into
$K$ disjoint blocks $B_1,\dots,B_K$, we define
\begin{equation}
  \label{eq:block_spin_latent}
  \tilde z_k(x)
  \;=\;
  \frac{1}{|B_k|}
  \sum_{i\in B_k} \mu_{\theta,i}(x),
  \qquad
  k = 1,\dots,K,
\end{equation}
and normalise if desired. For $K=3$ this defines an explicit map
$\mathcal{Z}\to\mathbb{R}^3$ which can be visualised.

Now, in our case, if the latent codes lie in a disordered phase, the coarse-grained
variables $\tilde z(x)$ should form a very diffuse, nearly isotropic
cloud in $\mathbb{R}^3$, since the coarse graining is highly sensitive to the specific configuration of different samples. In contrast, if the latent codes are in an ordered or partially ordered phase, then the cloud of block spins $\tilde z(x)$ should form one or a few tight clusters in $\mathbb{R}^3$ that are visually apparent (islands
on the hypersphere): the cluster assignments of datapoints should be \emph{stable} under coarse graining for different samples of a same clustering class. See Fig.\ref{renorm_diags}.

\subsection{Nested ordering in hyperspherical coordinates}

Classical spin-glass analyses on the sphere focus on order parameters
expressed in terms of overlaps $q(s,s')$. In our latent
construction we parametrise spins $s\in S^{N-1}$ by \emph{hyperspherical
coordinates} $(r,\phi_1,\ldots,\phi_{N-1})$ with $r\equiv\sqrt{N}$
and \emph{design the prior so as to control not just one angle but the
\textbf{entire vector of angles}}. This yields a hierarchy of
\emph{coordinate-wise order parameters}:
\begin{equation}
  m_k
  :=
  \mathbb{E}\bigl[\cos \phi_k\bigr],
  \;
  v_k
  :=
  \mathrm{Var}\bigl(\cos \phi_k\bigr),
  \;
  k=1,\dots,N-1.
\end{equation}
Each $(m_k,v_k)$ describes an ordering transition on the \emph{sub-}sphere
$\boxed{S^{N-1-k}}$ embedded in $S^{N-1}$. We call this $\underline{\text{\textbf{k-NOT}: \emph{k-Nested Order Transition}}}$. See \cite{Ascarate2026VAEHypersphericalAD} for implementation details.


\paragraph{Multi-step nested ordering.}
As we increase the compression strength or effective inverse
temperature, the spin-glass analogy predicts that different angular
coordinates may \emph{order at different scales}, leading to a
multi-step hierarchy. This \emph{nested} ordering pattern is different from the usual
Parisi replica-symmetry-breaking picture, where the hierarchy is
expressed in the structure of $P(q)$. Here, the hierarchy lives in
the marginal distributions of the hyperspherical coordinates
themselves. In our experiments we will track $(m_k,v_k)$ as functions
of the compression hyperparameter and identify distinct ``kinks'' or
plateaux corresponding to successive ordering transitions. See Fig.\ref{hnested_rep_angle_hists}.

\begin{figure}[!h]
    \centering
    \includegraphics[width=0.5\linewidth]{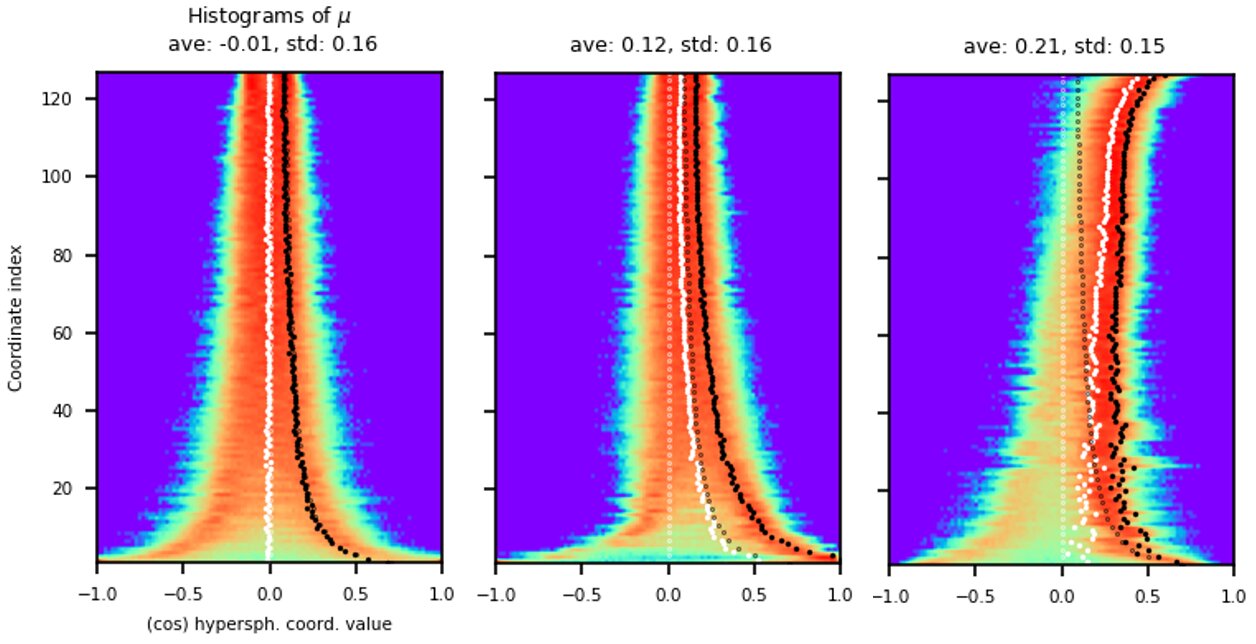}
    \caption{Each horizontal slice at some vertical index value shows the color coded histogram (red, high density; blue, low density) for the range of the coordinate of that index; the vertical axis stacks all the histograms for all the dimensions (in this example, $128$). The white dots represent the mean and the black dots represent the standard deviation of the corresponding histogram. The numbers on top are the total mean and standard deviation of all these previous values for all dimensions taken together. From left to right, the direction of the applied magnetic field in each experiment is shifted from the equator, to $(1,1,\dots,1)$, and then towards the north pole to bias the transition during training (first two, left to right, are the experiments of Fig.\ref{renorm_diags}).}
    \label{hnested_rep_angle_hists}
\end{figure}


\subsection{Susceptibility and critical behaviour}\label{P2}

In spin-glass systems, the response of the overlap or magnetisation
to changes in temperature or external field is captured by
susceptibilities. For example, the linear response of the
magnetisation $m$ to a small change in external field $h$ defines
the magnetic susceptibility $\chi = \partial m / \partial h$, which
typically diverges or peaks near critical points. In practice, one
often uses equivalent fluctuation-dissipation relations, where the
variance of the overlap or magnetisation under the Gibbs measure
plays the role of a susceptibility-like quantity.

In our latent context, we can define an \emph{overlap susceptibility}
associated to a control parameter $\alpha$ by
\begin{equation}
  \chi_{\mathrm{ovl}}(\alpha)
  \;:=\;
  \mathrm{Var}_{x,x'}\bigl[ R_\alpha(x,x') \bigr],
\end{equation}
where $R_\alpha$ is the latent overlap computed for a model trained
with hyperparameter $\alpha$. Heuristically, $\chi_{\mathrm{ovl}}$
should peak where the latent representation undergoes a qualitative
reorganisation (e.g.\ when islands appear or disappear). Indeed, as we vary
$\alpha$ from a weak-compression regime to a strong-compression
regime there is a distinguished value $\alpha^\star$ at which
$\chi_{\mathrm{ovl}}(\alpha)$ attains a maximum. We interpret
$\alpha^\star$ as an \emph{edge-of-stability} point for the latent
spin system: just before the topological structure of the energy
landscape simplifies (e.g.\ many small basins disappear), the
overlap fluctuations are maximal. Representation-quality metrics tend to peak near this same $\alpha^\star$. See Fig.\ref{AD_edge_stab} and \ref{ADGZ}.

\subsubsection{Topological trivialisation}

When an external field or other control parameter is increased, one
observes a \emph{topological trivialisation} transition: the energy
landscape goes from having many local minima (glassy phase) to
having only a single global minimum (trivial phase), and the
complexity collapses to zero (see \ref{topodiagra} for a figure and \citep{RosBenArousBiroliCammarota2019PRX} for more details). 
As we increase the compression strength or effective $\beta$, we
expect:
\begin{itemize}
  \item In a weak-compression regime: many shallow latent basins and
        fragmented low-energy sets.
  \item Near a critical value: a rapid drop in the estimated number
        of distinct low-energy basins.
  \item Beyond the critical value: an almost convex latent landscape
        with a single dominant basin (topologically trivial phase).
\end{itemize}

\subsubsection{Reaching the edge of topological trivialisation in practice}

When using the loss in hyperspherical coordinates \cite{Ascarate2025HypersphericalVAE_IJCNN}, we use an annealing schedule \citep{Fu2019CyclicalAnnealing} for the gain $\beta$ of the KLD-like loss, consisting of an initial stage which increases proportionally with \(\sqrt{\text{epoch}}\) for the first \(100\) epochs, and is constant afterwards. In our experiments on generation, for example, this was necessary because we observed that too much compression of the volume was detrimental to the performance (particularly for the MSE, while the self-FID instead still tends to improve), while a strong compression was still necessary at the initial stage (for a good self-FID), see \ref{gen} for an explicit full example from our runs. This shows true edge-of-stability/critical point behavior/regime. The total training was \(300\) epochs in all cases. Thus, when implementing this training process, and by a close monitoring of all of the previous diagnostic tools of spin glass origin, we were able to systematically shift our trained models towards the unstable region of interest, where the phase is at the edge of topological trivialization.


\section{Anomaly Detection and Generation Results}\label{results}

\subsection{Generation}
\label{subsec:generation-evidence}

The latent ordering picture is consistent with the generative improvements previously observed when the same hyperspherical compression mechanism is applied to VAEs. In \citet{Ascarate2025HypersphericalVAE_IJCNN}, the central difficulty is the standard high-dimensional VAE trade-off: increasing latent dimension improves reconstruction, but direct sampling from the Gaussian prior becomes progressively worse because the latent mass spreads over the high-volume equatorial region of the hypersphere. The proposed hyperspherical-coordinate loss compresses the latent representation into a lower-volume island on the hypersphere while preserving enough angular spread for reconstruction. This changes the reconstruction--generation trade-off rather than merely shifting one metric at the expense of the other.

On CIFAR-10, the compressed VAE improves the sampling quality by approximately \emph{ten self-FID points} relative to the best comparable standard VAE configuration, while simultaneously improving reconstruction by roughly \emph{one MSE point}. The self-FID was computed from $10{,}000$ generated samples against the reconstructed test distribution, so the metric specifically probes whether random latent samples decode into the learned data manifold rather than into visually meaningless regions. The same qualitative and quantitative pattern is reproduced on CelebA64: hyperspherical compression moves the latent away from the sparse equatorial regime and yields a simultaneous improvement in generation quality and reconstruction fidelity. Thus, the generative results give downstream evidence that the ordered latent phase is not merely a visualization artifact: the compressed island is dense enough to make random decoding meaningful, yet not so collapsed that reconstruction quality is destroyed, see \ref{gen}.

From the viewpoint of the present paper, these results support the interpretation that the useful regime is an intermediate ordered phase. A standard VAE leaves the latent close to a high-volume, almost-uniform hyperspherical distribution, which is geometrically sparse and poor for direct generation. Excessive compression would instead over-align the latent and harm reconstruction. The observed improvement occurs between these extremes: the latent codes become sufficiently ordered to form a dense sampling region, while retaining enough internal structure to reconstruct data accurately. This is precisely the type of edge-of-stability regime diagnosed by the overlap, susceptibility, and coarse-graining observables in the present work.

\subsection{Anomaly Detection}
\label{subsec:ad-evidence}

The same mechanism also produces substantial anomaly-detection gains in both fully unsupervised and conditional OOD settings. In \citet{Ascarate2026VAEHypersphericalAD}, anomalies are detected in the compressed latent space using a simple $k$-nearest-neighbor score on encoder means. The method is deliberately minimal: after training the VAE on nominal or ID data, a test point is encoded and scored by its average distance to the $k$ nearest training latents. The improvement therefore comes from the geometry of the latent representation itself, not from a specialized anomaly classifier.

In the fully unsupervised setting, the method was evaluated on two real-world image datasets: Mars Rover Mastcam and Galaxy Zoo 64. On Mars Rover, the compressed VAE improves AUROC by about \emph{10 percentage points} over the standard VAE baseline, rising from roughly $0.66$ for VAE+$k$NN to roughly $0.76$ for Comp.VAE+$k$NN. On Galaxy Zoo 64, it improves AUROC by about \emph{5 percentage points}, from roughly $0.74$ to roughly $0.79$. These datasets are important because they are not synthetic leave-one-class-out image benchmarks: Mars Rover contains unusual planetary imagery, while Galaxy Zoo uses human-labelled odd galaxies, making the anomaly notion closer to the intended real-world setting.

The OOD results show the same pattern at larger scale. With CIFAR-10 as ID, the compressed VAE was tested against standard far-OOD datasets including SVHN, Places365, LSUN-Crop, LSUN-Resize, iSUN, and Textures. In this far-OOD setting, the method lowers FPR95 by roughly \emph{ten points} relative to the relevant non-contrastive baselines, while remaining competitive in AUROC. The near-OOD setting is more stringent: CIFAR-100 is semantically close to CIFAR-10 and therefore lies closer to the ID support. There, the compressed VAE gives a much larger FPR95 improvement, around \emph{thirty points}, while retaining comparable AUROC. Finally, in the more complex near-OOD experiment with Imagenette as ID and semantically close ImageNet classes as OOD, the method again improves FPR95 by about \emph{ten points} over the implemented $k$NN baseline and also outperforms the vMF subcase obtained by compressing only one hyperspherical angle.

These AD gains are directly aligned with the latent spin-glass interpretation. In the standard VAE, both normal and anomalous samples can occupy a broad, high-volume latent region, making nearest-neighbor distances less discriminative. Hyperspherical compression instead creates a dense island for nominal data; anomalies tend to fall into shallower or more weakly populated regions of the latent landscape, increasing their distance to the nominal island. The fact that this works for both fully unsupervised real-world anomalies and near/far OOD benchmarks suggests that the ordered latent phase improves not only sample generation but also the separation geometry required for anomaly detection.

\begin{figure}[!h]
    \centering
    \includegraphics[width=1\linewidth]{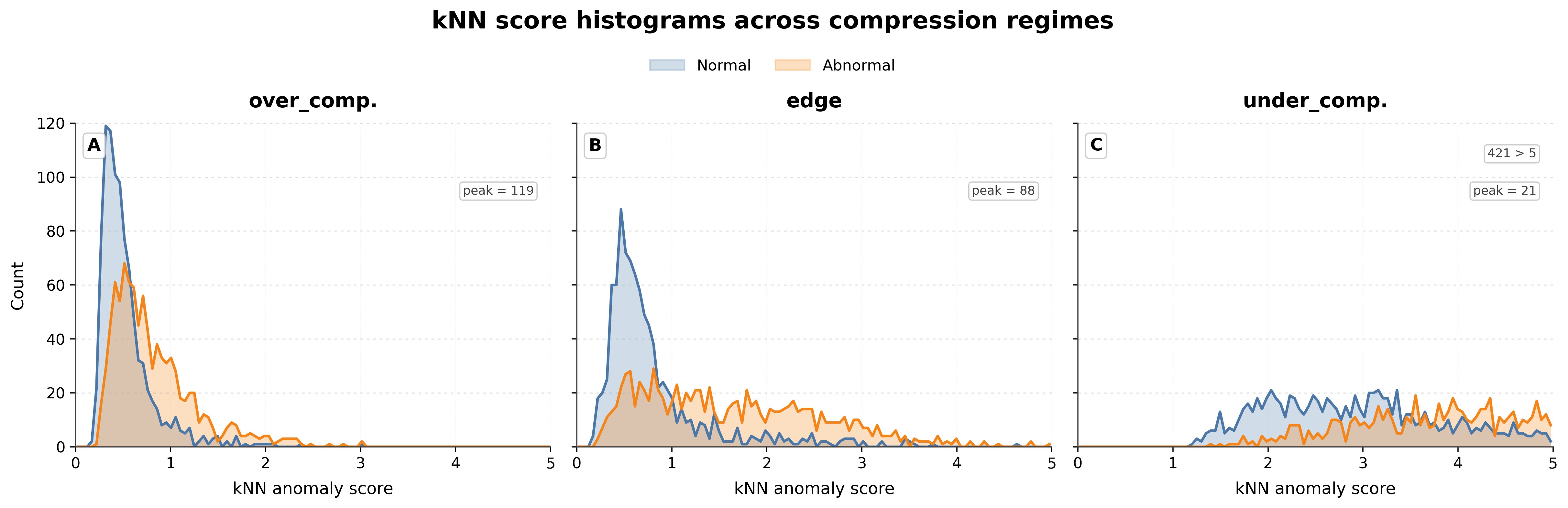}
    \caption{\textbf{At the edge of topological trivialization}. The $k-$NN method on latent space for anomaly detection is natural in this case, since it directly measures the Euclidean distances between replicas (cf.\ref{over}). In this way, the histograms for the score display quite directly the law of the overlap $P(R)$. At the edge of stability (\textbf{B}), both topologically trivial and continuous RSB phases are observed at the same time, the first in the normal data (class seen and volume-compressed during training of the VAE) and the second for the more energetic anomalies (\textit{not} seen in training). Left (\textbf{A}), when the intensity of the field is augmented, both classes topologically trivialize into a single peak; right (\textbf{C}), when it is decreased, both classes enter rugged RSB phases). This delicate instability of the configuration in (\textbf{B}) shows true edge-of-stability/near-critical point behavior. Experiments on the Galaxy Zoo dataset, which is an adequate dataset to check this since it contains anomalies varying in a continuous degree, from almost normal to very abnormal, so that the continuous full RSB region should be fully populated/covered by them. See \ref{ADGZ} for some other details.}
    \label{AD_edge_stab}
\end{figure}

\section{Conclusion}
The latent diagnostics suggest concrete knobs for practitioners: as argued in this paper, architectural choices (in particular, strength and direction of a latent prior) have \textit{predictable} signatures in overlaps and susceptibilities, offering a principled way to tune models toward regimes that are beneficial for representation quality.

\clearpage
\newpage

\bibliographystyle{plainnat}
\bibliography{references}

\newpage
\appendix

\section{Glossary of Spin-Glass Terms}
\label{sec:spin-glass-glossary}

\subsection*{---------------------------- 2D Ising Model and Block-Spin Renormalization Terms in Fig. \ref{renorm_diags} ----------}

\begin{description}[leftmargin=3.2cm, style=nextline]

\item[2D Ising model]
A classical lattice spin model in which each site $i$ of a two-dimensional grid carries a binary spin
\[
s_i\in\{-1,+1\}.
\]
The standard nearest-neighbor Hamiltonian is
\[
H(s)
=
-J\sum_{\langle i,j\rangle}s_i s_j
-h\sum_i s_i,
\]
where $J$ controls spin-spin coupling and $h$ is an external magnetic field.

\item[Square lattice]
The usual two-dimensional grid on which the 2D Ising model is defined. Each spin interacts with its nearest horizontal and vertical neighbors.

\item[Nearest neighbors]
Pairs of lattice sites $\langle i,j\rangle$ connected by one lattice edge. In the square-lattice Ising model, each interior spin has four nearest neighbors.

\item[Coupling constant $J$]
The parameter controlling whether neighboring spins prefer to align or anti-align. If $J>0$, aligned neighbors lower the energy and the model is ferromagnetic. If $J<0$, anti-aligned neighbors are favored and the model is antiferromagnetic.

\item[Ferromagnetic interaction]
A coupling that favors equal neighboring spins:
\[
s_i=s_j.
\]
In the 2D Ising model this corresponds to $J>0$.

\item[Antiferromagnetic interaction]
A coupling that favors opposite neighboring spins:
\[
s_i=-s_j.
\]
This corresponds to $J<0$ and can create alternating spin patterns.

\item[External field $h$]
A term biasing all spins toward one direction. Positive $h$ favors $s_i=+1$; negative $h$ favors $s_i=-1$.

\item[Magnetization]
The average spin,
\[
m=\frac{1}{N}\sum_{i=1}^N s_i.
\]
It is the standard order parameter for the ferromagnetic Ising transition. In a disordered phase, $m\approx 0$; in an ordered phase, $|m|>0$.

\item[Spin domain]
A connected region of the lattice where most spins have the same sign. Below the critical temperature, large domains appear and eventually dominate the lattice.

\item[Domain wall]
The boundary between neighboring regions of opposite spin. Domain walls are high-energy objects in the ferromagnetic Ising model because spins across the boundary are misaligned.

\item[Temperature $T$]
A parameter controlling thermal randomness. High $T$ makes spins fluctuate strongly and produces disorder; low $T$ suppresses fluctuations and favors ordered domains.

\item[Inverse temperature $\beta$]
The reciprocal temperature, usually $\beta=1/(k_B T)$ or with $k_B=1$, $\beta=1/T$. Increasing $\beta$ corresponds to cooling the system.

\item[Critical temperature $T_c$]
The temperature at which the 2D Ising model undergoes a continuous phase transition from a disordered phase to an ordered ferromagnetic phase.

\item[Critical point]
The point $T=T_c$ where the correlation length diverges, fluctuations occur at all scales, and the system becomes scale-invariant.

\item[Disordered phase]
The high-temperature phase $T>T_c$. Spins fluctuate almost independently, domains are small, and the average magnetization is close to zero.

\item[Ordered phase]
The low-temperature phase $T<T_c$. Large aligned domains form and the magnetization becomes nonzero in the thermodynamic limit.

\item[Correlation function]
A statistic measuring how strongly two spins are related as a function of their separation:
\[
C(r)=\langle s_i s_{i+r}\rangle-\langle s_i\rangle\langle s_{i+r}\rangle.
\]
It decays rapidly away from criticality and slowly at criticality.

\item[Correlation length]
The characteristic distance over which spin correlations decay. It is finite away from the critical point and diverges at $T_c$.

\item[Scale invariance]
A property at criticality where the system looks statistically similar across many length scales. This is why renormalization-group methods are natural near critical points.

\item[Coarse-graining]
A procedure that replaces many microscopic variables by fewer effective variables. In the Ising model, this means replacing groups of nearby spins by a block spin.

\item[Block]
A small local patch of the lattice, for example a $2\times 2$ or $3\times 3$ square of spins. RG transformations act by replacing each block with a single effective spin.

\item[Block spin]
The effective spin assigned to a block of microscopic spins. A common rule is majority vote:
\[
S_B=\operatorname{sign}\!\left(\sum_{i\in B}s_i\right),
\]
or an average followed by normalization. This is the Ising analogue of the latent block-averaging used in our paper.

\item[Majority rule]
A block-spin rule where the new spin takes the sign of the majority of spins inside the block. If the block contains more $+1$ than $-1$ spins, then $S_B=+1$; otherwise $S_B=-1$.

\item[Decimation]
A coarse-graining method where one removes part of the microscopic degrees of freedom, e.g. keeping only every other spin, and absorbs their effect into renormalized couplings.

\item[Renormalization-group transformation]
A map from one model description to another after coarse-graining:
\[
(J,h,T)\longmapsto (J',h',T').
\]
It describes how effective couplings change when short-scale degrees of freedom are integrated out or averaged.

\item[RG flow]
The trajectory obtained by repeatedly applying the renormalization-group transformation:
\[
(J,h,T)
\to
(J',h',T')
\to
(J'',h'',T'')
\to
\cdots.
\]
Different phases correspond to different long-run destinations of this flow.

\item[Fixed point]
A point in parameter space that is unchanged by the RG transformation. Critical points are usually associated with nontrivial fixed points, while ordered and disordered phases flow toward stable phase fixed points.

\item[Stable fixed point]
A fixed point that nearby RG trajectories flow toward. Stable fixed points represent phases, such as the ordered low-temperature phase or the disordered high-temperature phase.

\item[Unstable fixed point]
A fixed point with at least one direction in parameter space that flows away. The critical point is unstable along the temperature direction: a small move above or below $T_c$ flows to different phases.

\item[Relevant direction]
A perturbation that grows under RG iteration. Relevant directions determine what must be tuned to reach criticality, such as temperature in the 2D Ising model.

\item[Irrelevant direction]
A perturbation that shrinks under RG iteration. Irrelevant details disappear at large scales, explaining universality.

\item[Marginal direction]
A perturbation that is neither clearly growing nor shrinking at leading order. Marginal directions often require higher-order analysis.

\item[Universality]
The fact that many microscopically different systems share the same critical behavior. In RG language, they flow to the same critical fixed point.

\item[Universality class]
A family of systems with the same critical exponents and scaling laws. The 2D Ising universality class includes many systems with scalar $\mathbb{Z}_2$ symmetry and short-range interactions.

\item[Critical exponent]
A number describing how a physical quantity diverges or vanishes near criticality. For example,
\[
\xi_{\mathrm{corr}}\sim |T-T_c|^{-\nu}
\]
defines the correlation-length exponent $\nu$.

\item[Scaling law]
A power-law relation near criticality, such as the divergence of susceptibility or correlation length as $T\to T_c$.

\item[Susceptibility]
The response of magnetization to a small external field:
\[
\chi=\frac{\partial m}{\partial h}.
\]
It becomes large near the critical point because the system is highly sensitive to perturbations.

\item[Partition function]
The normalizing sum over all spin configurations:
\[
Z=\sum_{\{s_i\}}\exp[-\beta H(s)].
\]
It encodes the thermodynamics of the model.

\item[Free energy]
The thermodynamic potential
\[
F=-\frac{1}{\beta}\log Z.
\]
Non-analytic behavior of the free energy in the infinite-size limit signals a phase transition.

\item[Thermodynamic limit]
The limit where the number of spins $N$ goes to infinity. True phase transitions occur only in this limit; finite systems show rounded versions of the transition.

\item[Finite-size effects]
Deviations from infinite-system behavior due to finite lattice size. In finite lattices, susceptibilities peak but do not diverge.

\item[Finite-size scaling]
A method for extracting critical behavior by studying how observables change with system size.

\item[Real-space RG]
A renormalization method performed directly in physical space by grouping nearby spins into blocks. Block-spin transformations are a canonical real-space RG procedure.

\item[Length scale $L$]
The spatial resolution at which the system is described. Coarse-graining increases the effective length scale by replacing many microscopic spins with one effective spin.

\item[Scale factor $b$]
The factor by which length scale increases in one RG step. For example, grouping $2\times 2$ blocks in a square lattice corresponds to $b=2$.

\item[Effective Hamiltonian]
The Hamiltonian after coarse-graining. It has the form
\[
H'(S)=-J'\sum_{\langle B,B'\rangle}S_BS_{B'}-\cdots,
\]
where $S_B$ are block spins and the dots indicate that new interaction terms may be generated.

\item[Renormalized coupling]
A coupling parameter, such as $J'$, after one RG step. Even if the original model has only nearest-neighbor interactions, coarse-graining can generate longer-range and higher-order terms.

\item[Blocking]
The act of partitioning microscopic variables into local groups. In the paper's latent analogy, blocking corresponds to partitioning latent coordinates into groups and averaging them.

\item[Spin alignment]
The tendency of neighboring spins to take the same sign. Under block-spin RG, aligned regions remain stable and produce coherent block spins.

\item[Critical fluctuations]
Large fluctuations appearing near $T_c$ across many length scales. In block-spin pictures, these appear as domains of many different sizes.

\item[Self-similarity]
The visual and statistical recurrence of similar structures across scales. Critical Ising configurations are approximately self-similar.

\item[Latent block-spin analogue]
The paper's analogue of Ising block spins: group latent coordinates $B_k$ and form
\[
\tilde z_k(x)
=
\frac{1}{|B_k|}
\sum_{i\in B_k}\mu_{\theta,i}(x).
\]
This is not a physical-space block, but it is structurally analogous as a coarse-graining of microscopic coordinates into effective variables.

\item[Resolution, `Zoom in/out']
The level of detail retained after coarse-graining. In histogram or latent-space visualizations, changing the resolution is analogous to changing the scale at which phase structure is inspected.

\item[RG-style]
A qualifier meaning that we borrow the operational logic of renormalization---coarse-grain, rescale, and check stability of structure---without claiming a full exact RG transformation for the neural network.

\item[Coarse-grained cluster stability]
The persistence of latent clusters after block-spin averaging or after changing the blocking scheme. Stable clusters indicate ordered structure that survives reduction of degrees of freedom.

\end{description}

\subsection*{-------------------------------------- Spin-Glass Terms in General -------------------------------------------}

\begin{description}[leftmargin=2.8cm, style=nextline]

\item[Spin]
A microscopic degree of freedom of the system. In this paper, the analogue of a spin is a latent coordinate, or after normalization, a coordinate of the latent vector
$s(x)=\mu_\theta(x)/\|\mu_\theta(x)\|$.

\item[Spin configuration]
A full assignment of all spins. In our setting, a latent vector $s\in S^{N-1}$ is interpreted as a spin configuration.

\item[Hamiltonian]
The energy function of the system. Lower Hamiltonian values correspond to more likely or more stable configurations. In our latent setting, the Hamiltonian is the latent energy
\[
H_x(\mu)=\ell(x,D_{\theta^\star}(\mu))+\lambda \Phi(\mu),
\]
i.e., reconstruction cost plus latent prior.

\item[Latent energy]
The autoencoder-induced Hamiltonian on latent space. For a fixed decoder and input $x$, it measures how costly it is to represent $x$ by a latent code $\mu$.

\item[Spherical model]
A spin model where the spin vector is continuous and constrained to a sphere,
\[
s\in S^{N-1},\qquad \|s\|^2=N
\]
or, after normalization, $\|s\|=1$. This is the closest classical spin-glass analogue of a continuous VAE latent.

\item[$p$--spin interaction]
An interaction involving $p$ spins at once, e.g.
\[
J_{i_1\cdots i_p}s_{i_1}\cdots s_{i_p}.
\]
For $p=2$ this is a pairwise interaction; for $p\ge 3$ it produces increasingly high-order couplings and more rugged landscapes.

\item[Pure $p$--spin model]
A model whose Hamiltonian contains only one interaction order $p$. For example, a pure $3$--spin model contains only cubic spin products.

\item[Mixed $p$--spin model]
A model containing a sum of several interaction orders,
\[
H(s)=\sum_{p\ge 1}\beta_p H_p(s).
\]
The coefficients $\beta_p$ are the mixture weights. In the classical theory, these weights must decay sufficiently fast so that the infinite mixture is well-defined.

\item[Mixture profile $\xi$]
The real function
\[
\xi(t)=\sum_{p\ge 1}\beta_p^2 t^p
\]
that encodes the mixture weights of a mixed $p$--spin model. Functionally, the important regularity condition is that this series has radius of convergence strictly larger than one, equivalently that the high-$p$ weights decay geometrically.

\item[Mean-field scaling]
The normalization of random couplings with $N$ so that the total energy remains extensive, usually of order $N$. In mixed $p$--spin models, the $p$--spin couplings are scaled like $N^{-(p-1)/2}$.

\item[External magnetic field]
A bias that favors alignment of spins with a preferred direction. In our latent model, a prior centered at $m$,
\[
\Phi(\mu)=\frac12\|\mu-m\|^2,
\]
acts as an external field because, on the sphere, it contributes a linear term $-\langle m,\mu\rangle$.

\item[Gibbs measure]
The probability distribution over configurations at inverse temperature $\beta$:
\[
\pi_\beta(s)\propto \exp[-\beta H(s)].
\]
Low-energy states receive more probability mass. In our paper, this motivates viewing trained latent codes as samples biased toward low-energy regions.

\item[Inverse temperature $\beta$]
A parameter controlling concentration around low-energy states. Large $\beta$ means low temperature and strong concentration; small $\beta$ means high temperature and broad exploration.

\item[Langevin dynamics]
A noisy gradient flow of the form
\[
ds_t=-\nabla H(s_t)\,dt+\sqrt{2/\beta}\,dB_t.
\]
It combines descent in the energy landscape with thermal noise. We use it as an analogy for the noisy evolution of latent representations induced by SGD.

\item[Quenched disorder]
Randomness in the Hamiltonian that is treated as fixed while the system evolves. In neural networks, the trained or temporarily frozen decoder weights play the role of quasi-quenched disorder.

\item[Replica]
An independent sample from the same Gibbs distribution. In ML terms, one can think of two data points, or two latent codes drawn from the same trained representation distribution, as empirical replicas.

\item[Overlap]
The basic similarity measure between two spin configurations:
\[
q(s,s')=\frac{1}{N}\langle s,s'\rangle.
\]
In the paper, the latent overlap is
\[
R(x,x')=\frac{1}{N}\langle s(x),s(x')\rangle.
\]
It is analogous to cosine similarity after spherical normalization.

\item[Overlap distribution $P(q)$]
The distribution of overlaps between pairs of replicas. Its shape indicates the phase of the system: concentrated near zero in a disordered phase, shifted to positive values in an ordered phase, and possibly multi-peaked in glassy phases.

\item[Order parameter]
A statistic that diagnoses the phase of a system. Here, overlap distributions, overlap means/variances, and hyperspherical angular means act as latent-space order parameters.

\item[Disordered phase]
A high-temperature or weakly constrained phase where configurations are broadly spread. In latent space, this corresponds to codes spread almost uniformly on the hypersphere, with overlaps concentrated near zero.

\item[Ordered phase]
A phase where configurations align or cluster around preferred regions. In latent space, this corresponds to one or several compact islands with positive within-island overlap.

\item[Glassy phase]
A rugged phase with many low-energy states or basins. In ML language, this resembles a latent energy landscape with many local wells, fragmented low-energy regions, and nontrivial overlap structure.

\item[Replica symmetry]
A regime where replicas are statistically indistinguishable and the overlap structure is simple, often concentrated around a single value.

\item[Replica symmetry breaking (RSB)]
A regime where the overlap distribution becomes nontrivial, reflecting many organized pure states or basins. In this paper, RSB is used as the reference point for classical spin-glass hierarchy, contrasted with our nested hyperspherical ordering.

\item[Parisi hierarchy]
The classical hierarchical organization of overlaps in mean-field spin glasses. It is a hierarchy in the distribution of pairwise overlaps, not a hierarchy of individual hyperspherical coordinate marginals.

\item[Pure state]
A coherent cluster of configurations inside the Gibbs measure. In latent terms, a pure-state-like structure is analogous to a stable latent island or basin.

\item[Susceptibility]
A response measure that becomes large near a phase transition. In the paper, the overlap susceptibility
\[
\chi_{\mathrm{ovl}}(\alpha)=\mathrm{Var}_{x,x'}[R_\alpha(x,x')]
\]
is used as a finite-size indicator of critical reorganization in latent space.

\item[Critical point]
A value of a control parameter where the system changes phase or reorganizes sharply. In finite models this often appears as a peak rather than a true divergence.

\item[Edge of stability]
A marginal regime near a transition, where a stability mode is close to becoming unstable. Operationally, we associate this with peaks in susceptibility and, ideally, softening of Hessian modes or rapid changes in basin structure.

\item[Marginal stability]
A state where the smallest relevant stability eigenvalue is close to zero. In energy landscapes, this corresponds to soft directions and high sensitivity to perturbations.

\item[Complexity]
The exponential rate of growth in the number of critical points, often minima or saddles, at a given energy. In ML language, it measures how many distinct basins or stationary regions the landscape contains.

\item[Topological trivialization]
A transition where a rugged landscape with many minima becomes effectively single-basin or convex-like as a bias or field increases. In latent space, this corresponds to excessive compression, where most structure collapses into one dominant basin.

\item[Basin]
A region of the energy landscape attracted to the same local minimum under gradient descent. Latent islands can be interpreted as empirical basins of the induced latent energy.

\item[Hessian]
The matrix of second derivatives of the Hamiltonian. Its eigenvalues describe local curvature. Near marginal stability, small or near-zero eigenvalues indicate soft modes.

\item[Soft mode]
A direction in configuration space with very small curvature. Moving along a soft mode changes the energy only weakly, often indicating proximity to a stability boundary.

\item[Renormalization group (RG)]
A framework for studying how a system changes under coarse-graining. In this paper, ``RG-style'' means that we coarse-grain latent coordinates and track whether phase structure is stable under that reduction.

\item[Block-spin transformation]
A standard RG operation where groups of microscopic spins are averaged into effective spins. Our latent version averages blocks of latent coordinates:
\[
\tilde z_k(x)=\frac{1}{|B_k|}\sum_{i\in B_k}\mu_{\theta,i}(x).
\]

\item[Coarse-graining]
Reducing a high-dimensional system to fewer effective degrees of freedom while preserving large-scale structure. In the paper, coarse-graining maps high-dimensional latents to a low-dimensional block-spin representation.

\item[RG stability]
The persistence of cluster or phase structure under different coarse-grainings. If latent clusters remain visible after block averaging, we interpret this as evidence of ordered structure.

\item[Hyperspherical coordinates]
Coordinates on a high-dimensional sphere: radius plus angles
\[
(r,\phi_1,\ldots,\phi_{N-1}).
\]
They allow us to study ordering not only globally through overlaps, but coordinate-wise through angular distributions.

\item[Subsphere]
A lower-dimensional sphere obtained by fixing one or more angular coordinates, or equivalently by fixing an overlap with a reference direction.

\item[Band]
A thin region of the sphere at approximately fixed overlap with a reference configuration. Classical spherical spin-glass results show that restricting to such bands can again produce an effective mixed spherical model.

\item[Nested ordering]
The phenomenon where different hyperspherical angles order at different control-parameter values. This is distinct from Parisi RSB: the hierarchy is in coordinate-wise angular marginals, not in the global overlap distribution.

\item[$k$-Nested Order Transition ($k$-NOT)]
The paper's term for an ordering transition occurring at the $k$-th hyperspherical angular level. Each level corresponds to ordering on a successive subsphere.

\item[Control parameter]
The knob varied to induce a transition. In experiments this may be a compression strength, prior weight, KL coefficient, field direction, or effective inverse temperature.

\item[Compression]
In this paper, compression refers to the concentration of latent mass into smaller regions of the hypersphere. It is implemented through priors or hyperspherical-coordinate constraints that bias the latent toward ordered regions.

\end{description}

\subsection*{-------------------------------------- Other Terms used in the paper -------------------------------------------}

\begin{description}[leftmargin=3.4cm, style=nextline]

\item[Generated sample]
A decoded sample obtained by drawing a latent vector from a chosen latent sampling distribution and passing it through the decoder. In the standard VAE this is usually done by sampling from the Gaussian prior; in the compressed VAE, samples may be drawn from the empirically fitted compressed latent island.

\item[Reconstruction]
The decoder output obtained from the latent code of an input image:
\[
x \longmapsto \mu_\theta(x) \longmapsto \hat x = D_\theta(\mu_\theta(x)).
\]
Reconstruction quality measures whether the autoencoder preserves input information.

\item[Mean-squared error (MSE)]
A reconstruction metric,
\[
\mathrm{MSE}(x,\hat x)=\|x-\hat x\|_2^2,
\]
or its dataset average. Lower MSE means sharper or more faithful reconstructions. In the generation experiments, MSE is used together with self-FID to measure the reconstruction--generation trade-off.

\item[Fréchet Inception Distance (FID)]
A distributional image-quality metric comparing the feature distribution of generated images with that of a reference image set. Lower FID means the generated distribution is closer to the reference distribution in feature space.

\item[Self-FID]
The FID variant used in the generation paper. Instead of comparing generated images directly against the raw data distribution, self-FID compares randomly decoded/generated images against the reconstructed test distribution. Operationally, it asks whether latent samples decode into the manifold learned by the VAE. Lower self-FID means random latent samples produce outputs closer to the model's own reconstructed data distribution.

\item[Generation--reconstruction trade-off]
The empirical tension that increasing latent dimension or relaxing the prior can improve reconstruction but worsen random generation, because the latent space becomes too sparse for direct sampling. The desired regime is low MSE and low self-FID simultaneously.

\item[Latent sparsity]
The phenomenon that, in high dimension, training samples occupy a tiny fraction of the available latent volume. For a standard Gaussian VAE prior, samples concentrate near a high-dimensional hypersphere and are almost orthogonal, making direct prior sampling likely to hit regions unseen during training.

\item[Hypervolume compression]
The hyperspherical-coordinate mechanism that pushes latent codes away from high-volume equatorial regions and toward a lower-volume region of the hypersphere. The goal is to make the learned latent distribution denser and therefore more useful for generation and anomaly detection.

\item[Latent island]
A compact, dense region of the hypersphere where the compressed VAE places the training or ID data. In generation, sampling near this island improves decoded sample quality. In anomaly detection, distance from this island becomes a discriminative anomaly signal.

\item[Fully unsupervised anomaly detection]
An anomaly-detection setting in which the training data are all treated as one single nominal class and the model receives no information about sub-classes or internal structure inside the normal data. The model only learns ``normality'' as a single distribution. At test time, anomalies are detected by their deviation from this learned nominal latent island.

\item[Nominal data]
The data regarded as normal during training. In the fully unsupervised setting, all nominal samples are collapsed into a single normal class, without sub-class labels or class-conditional structure.

\item[Anomaly]
A test sample that does not belong to the nominal distribution. In the fully unsupervised experiments, examples include unusual Mars Rover images or odd galaxies in Galaxy Zoo.


\item[In-distribution (ID) data]
The normal data distribution in an OOD experiment. For example, CIFAR-10 can be used as ID, with OOD datasets such as SVHN, Places365, LSUN, iSUN, Textures, or CIFAR-100.

\item[Conditional OOD mode]
The OOD mode used when the normal/ID data have known sub-classes. Unlike the fully unsupervised case, the normal data are not treated as a single undifferentiated class: class information is used to form class-specific latent subclusters. The model compresses each ID class toward its own region of the hypersphere, and OOD samples are detected by their distance from these ID class islands.

\item[ID subcluster]
A class-specific component of the normal data in conditional OOD mode. For example, when CIFAR-10 is ID, each CIFAR-10 class defines one normal subcluster. These subclusters are all normal, but their labels provide internal structure that the model can exploit.

\item[Far-OOD]
An OOD setting where the anomalous samples come from datasets visually or semantically far from the ID data. Examples include CIFAR-10 as ID and SVHN, Places365, LSUN, iSUN, or Textures as OOD.

\item[Near-OOD]
An OOD setting where anomalous samples are semantically close to the ID data and therefore harder to separate. An example is CIFAR-10 as ID and CIFAR-100 as OOD.

\item[Complex near-OOD]
A stricter near-OOD setting where ID and OOD samples are semantically close and visually complex. In the AD paper, this corresponds to Imagenette as ID and selected close ImageNet classes as OOD.

\item[$k$-nearest-neighbor anomaly score]
A latent-space anomaly score defined by the average distance from a query latent point \(z\) to its \(k\) nearest nominal training latents:
\[
A(z)
=
\frac{1}{k}
\sum_{i=1}^{k} d(z,z_i),
\qquad
d(z,z_1)\le \cdots \le d(z,z_n).
\]
Larger values indicate that the query lies farther from the nominal latent island and is therefore more anomalous. In the AD experiments, \(k=3\) and \(d\) is the Euclidean distance in latent space.

\item[Latent \(k\)-NN]
The use of the \(k\)-NN anomaly score after encoding data into the VAE latent space. This differs from pixel-space \(k\)-NN, where distances are computed directly on raw images.

\item[Pixel-space \(k\)-NN]
A baseline anomaly detector that computes nearest-neighbor distances directly in image space. It does not use a learned latent representation.

\item[Isolation Forest (IF)]
A classical anomaly-detection baseline based on random partition trees. Points that are isolated quickly by random splits receive higher anomaly scores.

\item[Reconstruction-error anomaly score]
An anomaly score based on reconstruction failure:
\[
A(x)=\|x-\hat x\|.
\]
The intuition is that an autoencoder trained on normal data should reconstruct normal samples better than anomalies. In practice, this can fail when the autoencoder generalizes as a generic image compressor.

\item[AUROC]
Area under the receiver operating characteristic curve. It measures how well the anomaly score ranks anomalies above normal samples across all thresholds. Higher AUROC is better.

\item[FPR95]
False positive rate at \(95\%\) true positive rate. In OOD detection, it measures how many OOD/anomalous samples are incorrectly accepted as ID when the detector keeps \(95\%\) of ID samples. Lower FPR95 is better.

\item[True positive rate (TPR)]
The fraction of ID samples correctly accepted as ID under an OOD-detection convention, or equivalently the fraction of positives correctly detected under a chosen labelling convention. FPR95 fixes this value at \(95\%\) and reports the corresponding false-positive rate.

\item[False positive rate (FPR)]
The fraction of anomalous or OOD samples incorrectly classified as normal/ID. A low FPR means fewer anomalies are missed.

\item[vMF subcase]
A restricted version of hyperspherical compression in which only one angular coordinate is compressed, equivalent in spirit to using a single concentration parameter around a mean direction. The full compressed VAE is stronger because it compresses all hyperspherical angles, reducing hypervolume faster.

\item[Class-conditional hyperspherical compression]
The conditional-OOD version of the method. Each ID class is assigned a different reference direction or axis, and its latents are compressed toward the corresponding class island. This uses class labels of the normal data but no OOD labels for training.

\item[Single-island compression]
The fully unsupervised version of the method. All nominal training samples are compressed toward one common latent island because the model receives no information about subclasses inside the normal data.

\item[Latent distance to the island]
The geometric quantity underlying the \(k\)-NN anomaly score in the compressed latent space. Normal samples lie close to the island; anomalies tend to fall farther away, producing larger nearest-neighbor distances.

\item[Hyperspherical coordinate loss]
The modified VAE regularizer that applies constraints to the radius and angular coordinates of the latent distribution. It replaces the standard Cartesian KL geometry with a coordinate system adapted to high-dimensional spherical concentration.

\item[Angular compression]
The part of the hyperspherical-coordinate loss that moves angular coordinates away from high-volume equatorial regions. It is more effective than purely radial compression for reducing hyperspherical volume in high dimension.

\item[Radial constraint]
The part of the hyperspherical-coordinate loss that keeps the latent radius near the high-dimensional Gaussian shell, typically near \(\sqrt{N}\). This preserves the spherical geometry while angular compression controls where on the sphere the data lie.

\item[Equatorial regime]
The standard high-dimensional Gaussian/VAE regime in which latent mass lies near high-volume equatorial bands relative to many possible axes. This regime is sparse, almost-orthogonal, and problematic for direct generation and distance-based anomaly detection.

\item[Compressed regime]
The regime where the latent distribution has been moved away from the equator into a lower-volume region. This makes the training distribution denser and improves random generation or latent-distance anomaly detection.

\end{description}

\section{From autoencoders to spherical \texorpdfstring{$p$}{p}--spin glasses}
\label{sec:pspin_and_main_theorem}

In this section we recall the spherical $p$--spin glass Hamiltonian
(with external field), and then state our main theorem: an
autoencoder with a prior on the latent mean induces, for each fixed
decoder snapshot, a random Hamiltonian on the latent hypersphere
which is locally equivalent to a mixed $p$--spin glass with external
field. This formalises the spin-glass interpretation sketched in
Section~\ref{sec:encoder_spin_glass}.

\subsection{Spherical \texorpdfstring{$p$}{p}--spin glass with external field}

Let $N$ be the number of spins and let $S^{N-1} \;=\; \bigl\{ s \in \mathbb{R}^N : \|s\|_2^2 = N \bigr\}$ be the $(N-1)$--dimensional sphere of radius $\sqrt{N}$.
A \emph{pure $p$--spin spherical Hamiltonian} with external field is
a random function
\begin{equation}
  \label{eq:p_spin_pure}
  H^{(p)}_N(s)
  \;=\;
  - \sum_{1 \le i_1 < \cdots < i_p \le N}
    J^{(p)}_{i_1\cdots i_p}\, s_{i_1}\cdots s_{i_p}
  \;-\;
  \sum_{i=1}^N h_i\, s_i,
  \qquad s \in S^{N-1},
\end{equation}
where $\{J^{(p)}_{i_1\cdots i_p}\}$ are i.i.d.\ centred Gaussian
couplings with variance
$\mathbb{E}\bigl[(J^{(p)}_{i_1\cdots i_p})^2\bigr]
 = \frac{p!}{2 N^{p-1}}$,
and $\{h_i\}$ is an external field (which may be deterministic or
random) with typical magnitude of order $1$. The conventional
scaling by $N^{-(p-1)/2}$ ensures that $H^{(p)}_N(s)$ is of order
$N$ for typical configurations $s$.

The case $p=2$ corresponds to a (random) quadratic form with a
linear field term and exhibits a single convex phase. For $p\ge 3$,
even at $h_i = 0$ the landscape of $H^{(p)}_N$ is highly nonconvex,
with an exponential number of critical points and rich phase
behaviour; adding an external field can drive a \emph{topological
trivialisation} or \emph{convexification} transition in which most
local minima disappear and only one global basin remains.

More generally, one often considers \emph{mixed $p$--spin} models
obtained by superposing several pure $p$ terms:
\begin{equation}
  \label{eq:mixed_p_spin}
  H_N(s)
  \;=\;
  \sum_{p\ge 2} \alpha_p\, H^{(p)}_N(s)
  \;-\;
  \sum_{i=1}^N h_i\, s_i,
  \qquad
  \sum_{p\ge 2} \alpha_p^2 < \infty,
\end{equation}
where each $H^{(p)}_N$ is an independent pure $p$--spin Hamiltonian
with the scaling above and $\{\alpha_p\}$ are mixing coefficients.
The linear term
$\sum_i h_i s_i$ can be interpreted as a $p=1$ contribution; it acts
as a uniform external field which biases the spins towards a
preferred direction and, when strong enough, can destroy the glassy
structure of the landscape.

In the remainder we will compare the latent energy functions of
autoencoders to Hamiltonians of the form~\eqref{eq:mixed_p_spin} on
$S^{N-1}$.

\subsection{Autoencoder with a prior on the latent mean}\label{mfield}

We now specialise to a simple but representative autoencoder
architecture and introduce a prior on the latent mean. Let
$\mathcal{X} \subset \mathbb{R}^{d_x}$ be the input space and let
$\mathcal{Z} \subset \mathbb{R}^{d_z}$ be the latent space, with
$d_z = N$ large. We consider a deterministic encoder
$E_\theta : \mathcal{X}\to \mathcal{Z}$ and decoder
$D_\theta : \mathcal{Z}\to \mathcal{X}$,
parameterised by weights $\theta$. For each input $x$ we denote
\[
  \mu_\theta(x) \;:=\; E_\theta(x) \in \mathcal{Z},
\]
and interpret $\mu_\theta(x)$ as the latent mean (or simply the
latent code) associated with $x$.

We assume the training objective has the form
\begin{equation}
  \label{eq:AE_mu_prior_loss}
  \mathcal{L}(\theta)
  \;=\;
  \mathbb{E}_{x\sim\mathcal{D}}
  \Bigl[
    \ell\bigl(x, D_\theta(\mu_\theta(x))\bigr)
    \;+\;
    \lambda\, \Phi\bigl(\mu_\theta(x)\bigr)
  \Bigr],
\end{equation}
where $\mathcal{D}$ is the data distribution, $\ell$ is a
reconstruction loss (e.g.\ squared error), $\lambda>0$ is a
hyperparameter, and $\Phi:\mathcal{Z}\to\mathbb{R}$ is a prior term
on the latent mean. In this section we focus on a quadratic prior
centred at a preferred mean $m\in \mathcal{Z}$:
\begin{equation}
  \label{eq:mu_prior}
  \Phi(\mu)
  \;=\;
  \frac{1}{2}\,\|\mu - m\|_2^2
  \;=\;
  \frac{1}{2}\,\|\mu\|_2^2
  - \langle m,\mu\rangle
  + \text{const}.
\end{equation}
This corresponds to an isotropic Gaussian prior $\mathcal{N}(m, I)$
on the latent mean. For a fixed decoder snapshot
$\theta^\star$, we define the \emph{latent energy for input $x$} as
\begin{equation}
  \label{eq:latent_energy_mu}
  H_x(\mu)
  \;=\;
  \ell\bigl(x, D_{\theta^\star}(\mu)\bigr)
  \;+\;
  \lambda\, \Phi(\mu),
\end{equation}
with $\Phi$ given by~\eqref{eq:mu_prior}. For the purpose of this
section, $\mu$ is regarded as a free variable in $\mathcal{Z}$; the
encoder simply selects a particular configuration
$\mu_\theta(x)$ on which gradient descent is performed during
training.

To make contact with the spherical spin-glass framework, we further
restrict $\mu$ to lie on a hypersphere of radius $\sqrt{N}$:
\[
  \mu \in S^{N-1}
  \;=\;
  \bigl\{ \mu \in \mathbb{R}^N : \|\mu\|_2^2 = N \bigr\}.
\]
In the high-dimensional regime this is natural, since under a
Gaussian prior the norm $\|\mu\|_2$ concentrates near $\sqrt{N}$.
On $S^{N-1}$, the term $\tfrac{1}{2}\|\mu\|_2^2$ in
\eqref{eq:mu_prior} is constant and can be discarded; the prior
contributes an \emph{external field term} $-\lambda \langle m,\mu\rangle$
to $H_x$. Thus the prior on the latent mean plays precisely the role
of a magnetic field which biases the latent spins towards the
direction of $m$, we call this ``(hypervolume-)compression''.

\subsection{Latent manifold and induced Riemannian metric}

The previous discussion can be given a geometric formulation by
equipping the latent space with a natural Riemannian metric induced
by the encoder or by the training dynamics. Let $\mathcal{W}$ denote
the parameter space, equipped with the standard Euclidean metric
$d\mathcal{W}^2 = \langle d\theta, d\theta\rangle$. Consider a
smooth map
\[
  f : \mathcal{Z} \to \mathcal{W},
\]
which encodes how latent configurations $z$ are embedded into, or
represented in, weight space along training (for instance via the
Jacobian map or by associating to each $z$ the parameter vector
$\theta_t$ that produces it as $E_{\theta_t}(x)$ for some $x$). The
pullback of the Euclidean metric on $\mathcal{W}$ by $f$ defines a
Riemannian metric $g$ on the latent space:
\begin{equation}
  \label{eq:pullback_metric}
  g_z(u,v)
  \;=\;
  \bigl\langle Df_z(u),\, Df_z(v) \bigr\rangle_{\mathcal{W}},
  \qquad u,v \in T_z \mathcal{Z}.
\end{equation}
Thus, $(\mathcal{Z}, g)$ becomes a Riemannian manifold on which we
can consider gradient flows of scalar functions such as the latent
energy $H_x(\cdot;\theta)$. The training trajectory of the encoder
for a fixed $x$ can be viewed either as a path $\theta_t$ in weight
space or as a path $z_t(x)$ in latent space; the metric $g$ simply
translates the notion of steepest descent from $\mathcal{W}$ to
$\mathcal{Z}$.

Under this viewpoint, stochastic gradient descent on $\theta$ with
small step size induces, in the continuum limit, a stochastic
gradient flow on $(\mathcal{Z},g)$ with potential $H_x$ and effective
inverse temperature $\beta_{\text{eff}}$. As in the classical
spin-glass case, one expects the stationary measure of such a flow,
when it exists, to be close to a Gibbs distribution of the form
\begin{equation}
  \label{eq:gibbs_latent}
  \pi_{\beta_{\text{eff}}}(z \mid x)
  \;\propto\;
  \exp\bigl[-\beta_{\text{eff}} H_x(z; \theta)\bigr] \, d\mu_g(z),
\end{equation}
where $\mu_g$ is the Riemannian volume measure induced by $g$ on
$\mathcal{Z}$. In our applications we will often restrict attention
to latent codes normalized to lie on a hypersphere, in which case
$\mu_g$ reduces to the standard hyperspherical surface measure and
the analogy with spherical spin glasses becomes particularly sharp.

\medskip
In summary, we interpret the encoder network as defining, for each
input $x$, a disordered Hamiltonian $H_x(\cdot;\theta)$ on the
latent hypersphere, with the latent coordinates playing the role of
spins and the training dynamics mimicking Langevin evolution in this
energy landscape. This spin-glass interpretation justifies importing
tools such as overlap order parameters, phase transitions, and
renormalization-group transformations into the analysis of VAE
latents.

\subsection{Training dynamics and Langevin analogy}

Training the encoder--decoder pair is typically done by (stochastic)
gradient descent on the empirical loss
\[
  \mathcal{L}(\theta)
  \;=\;
  \mathbb{E}_{x\sim \mathcal{D}}\bigl[
    \ell\bigl(x, D_\theta(E_\theta(x))\bigr)
    + R\bigl(E_\theta(x)\bigr)
  \bigr],
\]
where $\mathcal{D}$ is the training data distribution. Let
$\theta_t$ denote the parameter vector after $t$ steps of (stochastic)
gradient descent with step size $\eta$. For a fixed input $x$, this
induces a time evolution of its latent representation
\[
  z_t(x) \;:=\; E_{\theta_t}(x) \in \mathcal{Z}.
\]
Thus, even though we are not explicitly updating $z$ during training,
every parameter update changes $z_\theta(x)$, and the trajectory
$t \mapsto z_t(x)$ can be thought of as the dynamical evolution of
a spin configuration under a time-dependent Hamiltonian.

To make the analogy more precise, consider a single gradient step
on $\theta$:
\[
  \theta_{t+1} \;=\; \theta_t - \eta \,\nabla_\theta \mathcal{L}(\theta_t)
  \;+\; \text{(stochastic noise)}.
\]
By the chain rule, the induced change in $z_t(x)$ is
\[
  z_{t+1}(x) - z_t(x)
  \;\approx\;
  J_\theta(x)\, \bigl(\theta_{t+1} - \theta_t\bigr),
\]
where $J_\theta(x) = \nabla_\theta E_\theta(x)$ is the Jacobian of
the encoder with respect to the parameters. Substituting the update
on $\theta$ yields a discrete-time stochastic dynamical system for
$z_t(x)$ of the schematic form
\begin{equation}
  \label{eq:z_dynamics}
  z_{t+1}(x)
  \;\approx\;
  z_t(x)
  \;-\;
  \eta\, J_\theta(x)\,\nabla_\theta \mathcal{L}(\theta_t)
  \;+\;
  \text{(projected noise)}.
\end{equation}
On the other hand, the gradient of the latent energy
\eqref{eq:latent_energy} with respect to $z$ is
\[
  \nabla_z H_x(z;\theta)
  \;=\;
  \nabla_z \ell\bigl(x, D_\theta(z)\bigr)
  + \nabla_z R(z),
\]
so that, after appropriate reparametrization of time and under
standard smoothness assumptions, the induced dynamics of $z_t(x)$
can be viewed as a stochastic gradient flow on $H_x(\cdot;\theta_t)$:
\begin{equation}
  \label{eq:langevin_latent}
  z_{t+1}(x)
  \;\approx\;
  z_t(x)
  \;-\; \eta_{\text{eff}}\, \nabla_z H_x\bigl(z_t(x); \theta_t\bigr)
  \;+\; \sqrt{2/\beta_{\text{eff}}}\,\xi_t,
\end{equation}
where $\xi_t$ models the noise coming from stochastic gradients and
mini-batching, and $\eta_{\text{eff}}$, $\beta_{\text{eff}}$ are
effective step size and inverse temperature parameters determined by
the learning rate, the Jacobian $J_\theta(x)$, and the variance of
the gradient noise.

Equation~\eqref{eq:langevin_latent} is the latent-space analogue of
Langevin dynamics~\eqref{eq:langevin}, with the Hamiltonian given by
the latent energy $H_x$. In this picture,
\begin{itemize}
  \item the latent coordinates $z_i$ are continuous spins;
  \item the changing weights $\theta_t$ provide a slowly evolving
        source of quenched disorder in $H_x(\cdot;\theta_t)$;
  \item the encoder trajectory $z_t(x)$ is a noisy gradient descent
        path on the latent energy landscape;
  \item the effective inverse temperature $\beta_{\text{eff}}$ is
        inversely related to the learning rate and to the strength
        of gradient noise.
\end{itemize}
Intuitively, large learning rates and strong noise correspond to
high temperature (exploration of many latent configurations), while
small learning rates and weak noise correspond to low temperature
(exploration restricted to deep basins).

\subsection{A mixed $p$--spin expansion of the latent energy}

To formalise the spin-glass analogy we need slightly more structure
than the bare Taylor theorem. In particular, we want the latent
Hamiltonian $H_x$ to admit an expansion in homogeneous polynomials
whose coefficients decay sufficiently fast with the degree $p$ so
that the mixture is well-defined in the usual sense of mixed
$p$--spin models.

We therefore introduce the following notion.


\begin{definition}[Mixed spherical $p$--spin type with exponential mixture]
\label{def:mixed_p_spin_type}
A random Hamiltonian $H_N : S^{N-1}\to\mathbb{R}$ is of
\emph{mixed spherical $p$--spin type} if it admits the expansion
\begin{equation}
\label{eq:mixed_p_spin_type}
H_N(s)
\;=\;
C_N
\;-\;
\sum_{p\ge 1}
\sum_{1\le i_1<\cdots<i_p\le N}
J^{(p)}_{i_1\cdots i_p}\, s_{i_1}\cdots s_{i_p},
\qquad s\in S^{N-1},
\end{equation}
where $C_N$ is (possibly random) and, for each $p\ge 1$, the symmetric
tensor $J^{(p)}=(J^{(p)}_{i_1\cdots i_p})$ satisfies the \emph{mean-field scaling}
\begin{equation}
\label{eq:variance_scaling}
\bigl\|J^{(p)}\bigr\|_{\mathrm{F}}^2
\;=\;
\sum_{1\le i_1<\cdots<i_p\le N}
\mathbb{E}\!\left[(J^{(p)}_{i_1\cdots i_p})^2\right]
\;\le\;
\frac{\sigma_p^2}{N^{\,p-1}},
\end{equation}
for a coefficient sequence $(\sigma_p)_{p\ge1}$ with \emph{exponential decay}:
there exists $\rho\in(0,1)$ and $A<\infty$ such that
\begin{equation}
\label{eq:exp_decay_sigma}
\sigma_p \;\le\; A\,\rho^{\,p}
\qquad\text{for all }p\ge1,
\end{equation}
equivalently $\sum_{p\ge1}(1+\varepsilon)^{p}\sigma_p^2<\infty$ for some $\varepsilon>0$.
The linear term ($p=1$) acts as an external field; $p\ge2$ encode multi-spin interactions.
\end{definition}

\begin{remark}[Literature convention]
The exponential mixture condition \eqref{eq:exp_decay_sigma} is the standard
regularity assumption in rigorous treatments of mixed $p$--spin models,
e.g.\ requiring $\sum_{p\ge1}2^{p}\beta_p^{2}<\infty$ or, equivalently, that
the covariance $\xi(t)=\sum_{p\ge1}\beta_p^2 t^p$ extends analytically
beyond $t=1$; see, for instance, \cite{Panchenko2013} (Thm.\ 1.2) and
\cite{JagannathTobasco2015} (Assump.\ 1.1).
\end{remark}

Note that Definition~\ref{def:mixed_p_spin_type} does \emph{not}
assert that the couplings are independent or exactly Gaussian; it
only enforces the characteristic norm scaling with $N$ and the
decay of the mixture coefficients with $p$. In the wide-network,
random-weights limit, one expects approximate Gaussianity and weak
dependence by central limit arguments, but we will not need to make
that precise for our purposes.

We can now state our main structural result for the autoencoder
latent energy.

\begin{theorem}[Latent mixed-degree expansion with exponential \emph{decay} in $p$]
\label{thm:mixed_p_spin_latent}
Fix an input $x$ and define the latent energy
\[
H_x(\mu)\;=\;\ell\!\bigl(x, D_{\theta^\star}(\mu)\bigr)\;+\;\lambda\,\tfrac12\|\mu-m\|_2^2,
\qquad \mu\in\mathbb{R}^N,
\]
where $D_{\theta^\star}$ is a fixed decoder and $\ell$ a reconstruction loss.

\paragraph{Assumption (exponential derivative \emph{decay} on a latent neighbourhood).}
There exist a compact set $K\subset\mathbb{R}^N$ containing all latent codes of interest, constants $C_x<\infty$ and $\rho_x\in(0,1)$ such that for all $p\ge1$,
\begin{equation}
\label{eq:decay-assumption}
\sup_{\mu\in K}\Bigl\|\nabla^{(p)}_{\!\mu}\bigl[\ell\!\left(x, D_{\theta^\star}(\mu)\right)\bigr]\Bigr\|
\;\le\; C_x\,\rho_x^{\,p}\,p!.
\end{equation}
(The prior contributes only degrees $p\le2$ and can be absorbed into constants.)

Let $s=\mu/\sqrt{N}\in S^{N-1}$ and expand $H_x(\sqrt{N}\,s)$ homogeneously:
\begin{equation}
\label{eq:latent-mixed-expansion-clean}
H_x(\sqrt{N}\,s)
\;=\;
C_x
\;-\;
\sum_{p\ge1}\;
\sum_{1\le i_1<\cdots<i_p\le N}
J^{(p)}_{x,i_1\cdots i_p}\,s_{i_1}\cdots s_{i_p}.
\end{equation}
Then there exist symmetric coefficient tensors $J^{(p)}_x$ such that, with the mean-field normalization
\(
\widetilde J^{(p)}_x := N^{-(p-1)/2}\,J^{(p)}_x,
\)
one has the exponential \emph{decay} bound
\begin{equation}
\label{eq:mf-exp-decay}
\bigl\|\widetilde J^{(p)}_x\bigr\|_{\!F}
\;\le\; A_x\,\rho_x^{\,p}
\qquad\text{for all }p\ge1,
\end{equation}
for some constant $A_x<\infty$ (independent of $p$). Equivalently, setting $\sigma_p(x):=A_x\rho_x^{\,p}$,
\[
\|J^{(p)}_x\|_{\!F}\;\le\;\sigma_p(x)\,N^{(p-1)/2}
\quad\text{with}\quad
\sum_{p\ge1}(1+\varepsilon)^{p}\,\sigma_p(x)^2<\infty\ \text{for some }\varepsilon>0.
\]
The linear part is $-\sum_{i=1}^N J^{(1)}_{x,i}\,s_i=-\langle h_x,s\rangle$, where $h_x$ combines the prior field $-\lambda m$ and the linear Taylor term of $\ell(x,D_{\theta^\star}(\cdot))$ on $K$.
\end{theorem}

\begin{proof}[Sketch]
Fix $x$ and work on the compact $K$. By \eqref{eq:decay-assumption} and adding the (at most quadratic) prior, the $p$-th derivative tensors of $H_x$ satisfy
\(
\sup_{\mu\in K}\|\nabla^{(p)}_{\!\mu}H_x(\mu)\|\le C'_x\,\rho_x^{\,p}\,p!
\)
for all $p\ge1$. Taylor-expanding $H_x$ and substituting $\mu=\sqrt{N}\,s$ yield a homogeneous expansion in the monomials $s_{i_1}\cdots s_{i_p}$. Symmetrization gives coefficient tensors
\(
J^{(p)}_x \propto (\sqrt{N})^{p}\,\nabla^{(p)}_{\!\mu}H_x/p!
\)
up to universal combinatorial factors. Hence
\[
\|J^{(p)}_x\|_{\!F}
\;\lesssim\;
(\sqrt{N})^{p}\,C'_x\,\rho_x^{\,p}.
\]
Divide by the mean-field factor $N^{(p-1)/2}$ to obtain
\[
\bigl\|\widetilde J^{(p)}_x\bigr\|_{\!F}
\;=\;N^{-(p-1)/2}\,\|J^{(p)}_x\|_{\!F}
\;\lesssim\;C'_x\,\rho_x^{\,p}\,N^{1/2}.
\]
Since $N$ is fixed in the statement, absorb $N^{1/2}$ and the (dimension-free) combinatorial constants into $A_x$, which yields \eqref{eq:mf-exp-decay}. On the sphere, the quadratic prior becomes a constant; its linear part produces the field $h_x$, giving the stated form of the $p=1$ term.
\end{proof}

\paragraph{Analyticity vs.\ the derivative–decay assumption.}
If we were willing to assume that both the decoder $D_{\theta^\star}$ and the reconstruction map $\mu\mapsto\ell(x,D_{\theta^\star}(\mu))$ extend analytically to a complex neighbourhood of the latent region $K$, then the exponential derivative bound in Assumption~\ref{thm:mixed_p_spin_latent} would be automatic: by standard Cauchy estimates there exist $M_x<\infty$ and a radius $R_x>0$ such that $\sup_{\mu\in K}\|\nabla_\mu^{(p)}\ell(x,D_{\theta^\star}(\mu))\|\le M_x\,R_x^{-p}p!$ for all $p\ge1$, i.e.\ the assumption holds with $\rho_x=R_x^{-1}\!<1$. We \emph{do not} impose analyticity, however, because we want to include common non-analytic architectures (e.g.\ ReLU, max-pooling, hard-thresholding, certain normalisation layers) where higher derivatives are only piecewise defined and Cauchy bounds do not apply. Our theorem is therefore stated under the weaker, local \emph{exponential derivative–decay} condition \eqref{eq:decay-assumption}, which can also be empirically verified by (P0) in \ref{P0}. In practice, one can recover the analytic regime either by using smooth activations (e.g.\ SiLU/GELU/softplus) or by a mild Gaussian mollification of the latent map (equivalently, adding small latent noise, which we do in \ref{P0}), but we keep the assumption minimal precisely to cover ReLU-style decoders without additional smoothing.

In the wide-network, random-weights regime, one expects the entries
of $J^{(p)}_x$ to be approximately Gaussian and weakly dependent as
$N\to\infty$, by central limit heuristics for multilinear forms of
independent weights. We will use this as guiding intuition, but the
analysis in the rest of the paper relies only on the mixed
$p$--spin-type structure of $H_x$ as encoded in
Definition~\ref{def:mixed_p_spin_type} and
Theorem~\ref{thm:mixed_p_spin_latent}.

\subsection{(P0) Taylor–spin–glass coefficient decay check.}\label{P0}
\label{par:P0_decay}
We empirically verify the key assumption that the mixed-$p$ expansion coefficients of the latent Hamiltonian decay (effectively) exponentially in $p$ after the hyperspherical normalisation. We do this by directly measuring higher-order directional derivatives of $H_x(\mu)$ at the realised codes $\mu_\theta(x)$ up to a finite order $P_{\max}$.

\textbf{Object of interest.}
Let $H_x(\mu)=\ell(x,D(\mu))+\lambda\,\Phi(\mu)$ and fix $x$ and a unit direction $v\in\mathbb{R}^N$. Define the $p$-th directional derivative at $\mu=\mu_\theta(x)$:
\[
D_p(x;v) \;:=\; \left.\frac{d^p}{dt^p} H_x\bigl(\mu_\theta(x) + t\,v\bigr)\right|_{t=0}.
\]
The Taylor coefficient (scalar) along $v$ is $a_p(x;v) := \frac{1}{p!}\,D_p(x;v)$. To compare with the spherical spin-glass scaling on $s=\mu/\sqrt{N}$, we normalise by $(\sqrt{N})^p$:
\[
\widehat{a}_p(x;v) \;:=\; \frac{1}{(\sqrt{N})^p}\,\frac{1}{p!}\,D_p(x;v).
\]
For $p\ge 3$, only the decoder–loss composition contributes (the Gaussian $\Phi$ contributes at most $p\le 2$), so we fit decay on $p\in\{3,\dots,P_{\max}\}$.

\textbf{Protocol.}
\begin{enumerate}[leftmargin=1.25em]
  \item Select $P_{\max}\in\{6,8\}$, a held-out subset $X_{\text{sub}}$ (size $m=512$), and per-$x$ direction count $n_{\text{dir}}=16$. Sample $v$ i.i.d.\ from $\mathcal{N}(0,I/N)$ and renormalise to $\|v\|_2=1$.
  \item Compute $D_p(x;v)$ via automatic differentiation. Use double precision and clip $|D_p|$ at the $99.9$-th percentile to reduce outlier blow-up.
  \item Form $\widehat{a}_p(x;v)$ and aggregate per $\alpha$ (the control parameter) by
  \[
    \widehat{c}_p(\alpha) \;:=\; \operatorname{median}_{x\in X_{\text{sub}},\,v}\bigl|\widehat{a}_p(x;v)\bigr|.
  \]
\end{enumerate}

\textbf{Exponential-decay test.}
We test whether $\widehat{c}_p(\alpha)$ decays (approximately) geometrically in $p$:

\vspace{2pt}
\hspace{1.25em} \emph{Ratio test:} $r_p(\alpha):=\widehat{c}_{p+1}(\alpha)/\widehat{c}_p(\alpha)$ for $p=3,\dots,P_{\max}-1$. Report the geometric mean $\bar r(\alpha)$ and its bootstrap $95\%$ CI over seeds. Passing criterion: $\bar r(\alpha)<1$ with CI not crossing $1$. See Fig.\ref{decay} for results.

\begin{figure}[h]
    \centering
    \includegraphics[width=0.9\linewidth]{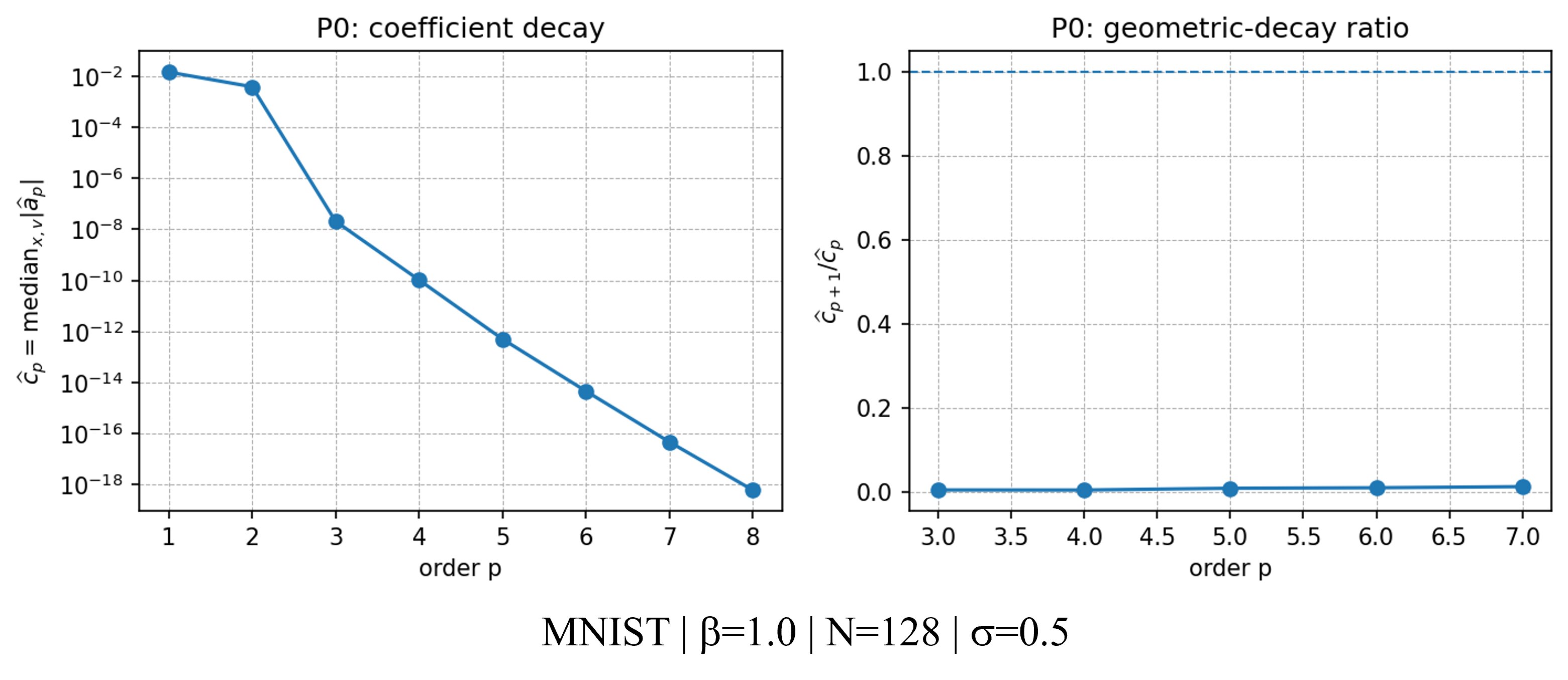}
    \caption{Decay test ($\sigma$ refers to the variance of the noise for the mollification).}
    \label{decay}
\end{figure}

\section{Topological trivialization diagramatics}\label{topodiagra}

\begin{figure}[!h]
    \centering
    \includegraphics[width=1\linewidth]{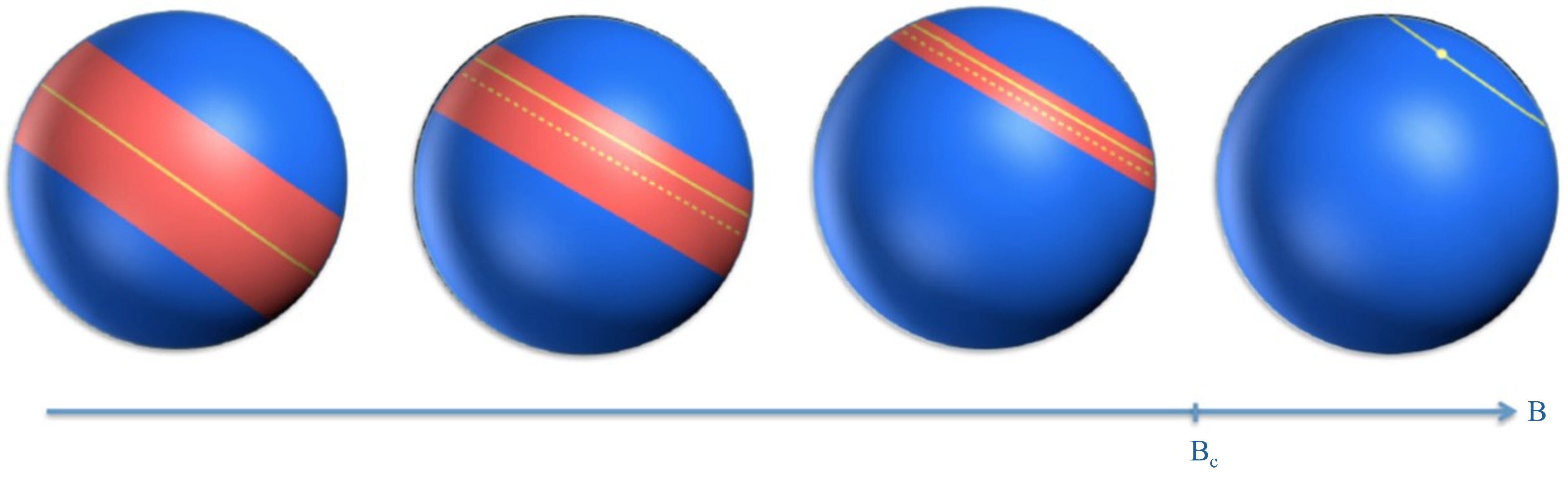}
    \caption{Evolution of the energy landscape in terms of the applied external magnetic field of intensity $B$. The red strip denotes the region on the sphere where minima lie in an exponential number. The continuous yellow line corresponds to the parallel where the deepest minima are located. The dashed yellow line corresponds to the parallel where the most numerous minima are located. At a critical field value $B_c$, the energy landscape has a transition: For $B< B_c$, it is rough and full of minima; for $B_c<B$, it is smooth and contains only one minimum (represented by the yellow dot in the figure). This (static) phase transition at $B_c$ is called ‘topological trivialization’ (of the landscape), a ‘convexification’ of the landscape. Image taken and adapted from \citep{RosBenArousBiroliCammarota2019PRX}.}
    \label{fig:placeholder}
\end{figure}

\begin{figure}[h]
    \centering
    \includegraphics[width=1.0\linewidth]{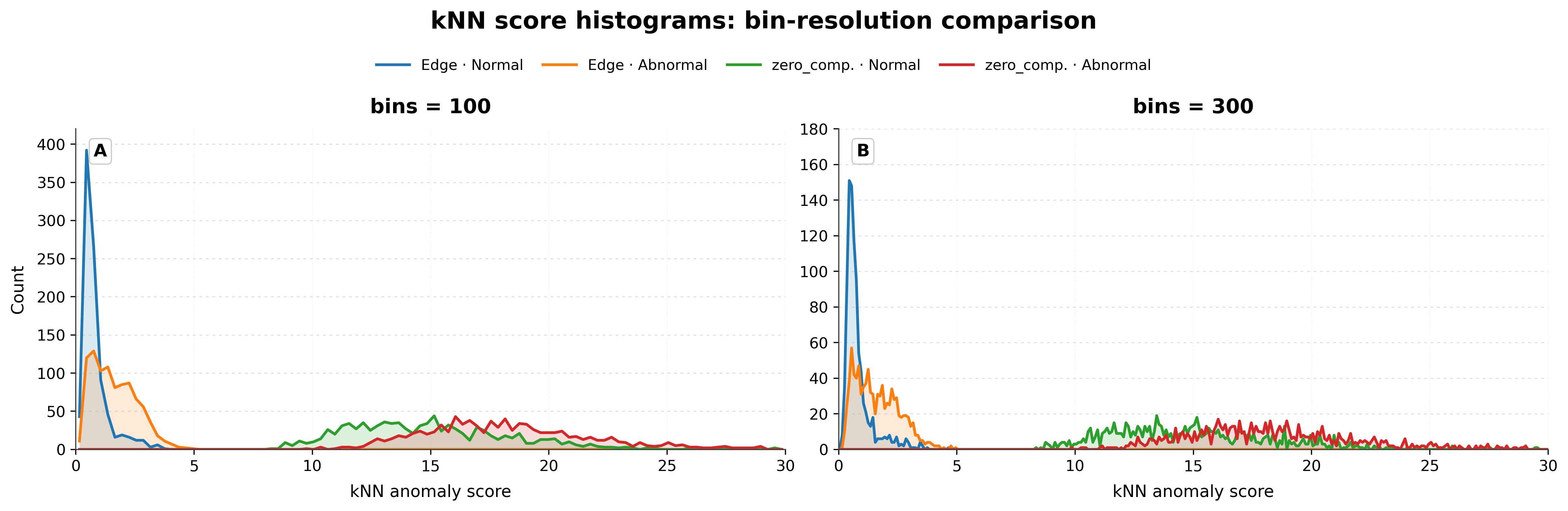}
    \caption{A full RSB region should be full of peaks around different, but very close, values of distance when zooming in, while a topologically trivialized region (single peak) should remain stable, and away from the full RSB region. To test if the above description is correct, we can do a renormalization group transformation on these histograms, modifying the resolution by changing the bins parameter. Same example as that of Fig. \ref{AD_edge_stab}.}
    \label{fig:AD_edge_stab}
\end{figure}

\clearpage
\newpage

\section{The spin glass analogy during training for generation}\label{gen}

The quality of randomly generated data can be measured using the Frechet Inception Distance (FID) \cite{Heusel2017TTUR}. FID compares the distribution of features between the images of the training/testing dataset and an equivalent number of randomly generated images.  We used in this experiment CIFAR10 \cite{Krizhevsky2009TinyImages} and an FID computed using 10,000 samples (we compare the random decoded samples with the \textit{reconstructed} testing set). We call this way of measuring the generation as `self-FID'. A good VAE-based generative model should minimize both of these metrics \textbf{simultaneously}: that is, to be able to generate random samples which are \textbf{in-distribution} w.r.t. the reconstructed dataset (low self-FID), and such that the latter actually resembles the original dataset (low MSE). 

When using the loss in hyperspherical coordinates \cite{Ascarate2025HypersphericalVAE_IJCNN}, we use an annealing schedule \citep{Fu2019CyclicalAnnealing} for the gain $\beta$ of the KLD-like loss, consisting of an initial stage which increases proportionally with \(\sqrt{\text{epoch}}\) for the first \(100\) epochs, and is constant afterwards. This was necessary because we observed that too much compression of the volume was detrimental to the performance (particularly for the MSE, while the self-FID instead still tends to improve), while a strong compression was still necessary at the initial stage (for a good self-FID). This shows true edge-of-stability/critical point behavior/regime. The total training was \(300\) epochs in all cases.  

The gain $\beta$ here has the role of the inverse temperature, $\beta=1/T$. In spin glasses and complex systems, the energy function has exponentially many local minima in the equatorial region of the hypersphere. To overcome them, a very strong signal or bias towards the desired region is necessary at the beginning, together with a rapid cooling or quenching. Thus, our initial high $\beta$ (i.e., very low temperature $T$) setting, and in the presence of the high intensity (regulated by the $\beta^{-1}$ factor in front of the MSE) hyperspherical external magnetic fields as bias in directions away from the equator, should make the gradient descent dynamics to quickly tend towards a low temperature distribution with replica symmetry breaking. Indeed, this is what we observed in our experiments, since we check for the replica angle, as mentioned before. This initial strong compression helps escaping those undesirable equatorial minima (Fig.\ref{fig:16}). Nevertheless, the obtained state shows too much overlapping between samples, so we then perform the annealing (i.e., lower the $\beta$, or increase the temperature $T$, and also lower the intensity of the magnetic fields) in order to allow the system to relax the strong order introduced by the initial bias and, in this way, transition to a replica symmetry breaking state with a bigger angle between replicas (that is, to go back up a bit in the ultrametricity tree/hierarchy of the replica angle values. This decreases the MSE and makes the decoded images more sharp, at the cost of some generation quality (Fig.\ref{fig:17}). Note how the replica angle (red dashed lines in fourth histogram to the left in second row) doesn't fully go back to $\pi/2$, even when the KLD term (where the external magnetic fields are) stops optimizing at this stage of the training process (red line in third row), but instead jumps to a different value, higher than the initial one but still below $\pi/2$. This is fully consistent with the spin glass analogy in a quenched and then annealed system, where the glass, always in the replica symmetry breaking phase, jumps from one so-called `pure state' to a different pure state, i.e., goes back up a bit in the ultrametricity tree/hierarchy of the replica angle values, as mentioned before. But the system has escaped the zone with exponentially many local minima in the equator. 

For more details on improvements in generative metrics in this regime, see \cite{Ascarate2025HypersphericalVAE_IJCNN}.

\begin{figure}[!h]
    \centering
    \includegraphics[width=0.8\linewidth]{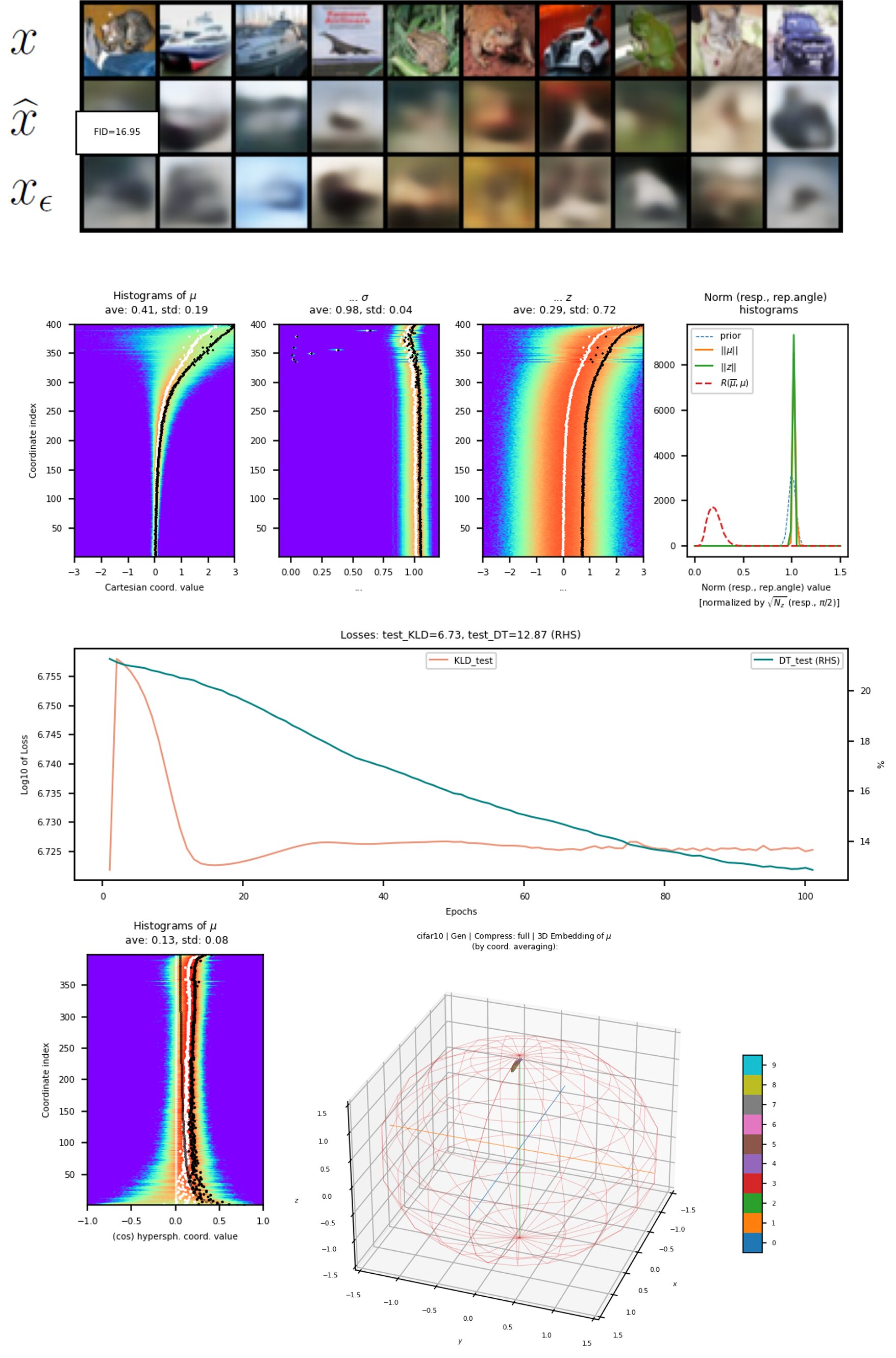}
    \caption{Results of a typical fully compressed VAE training at epoch $100$. The top panel shows the original data ($x$), the reconstruction ($\hat{x}$), and the generation sampling from the prior ($x_\epsilon$). The third panel/row shows the behavior of the MSE (green) and KLD (red) losses during training for the test set. Replica angle: red dashed lines in histogram to the left in second row.}
    \label{fig:16}
\end{figure}

\begin{figure}[!h]
    \centering
    \includegraphics[width=0.8\linewidth]{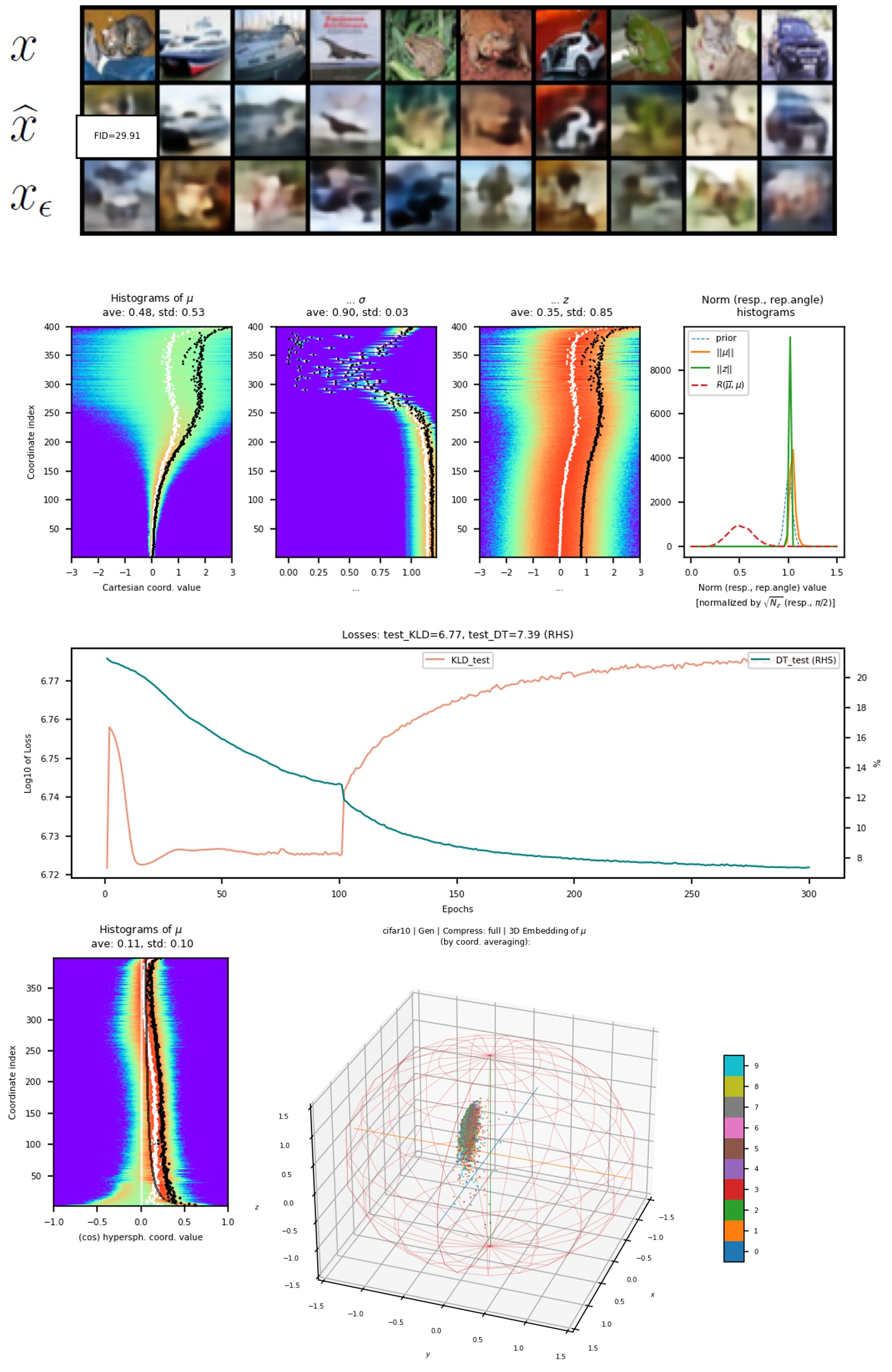}
    \caption{Results of the same fully compressed VAE training at final epoch $300$.}
    \label{fig:17}
\end{figure}

\newpage

\clearpage

\section{AD metrics improvement in GZ}\label{ADGZ}

For more details, see \cite{Ascarate2026VAEHypersphericalAD}.

\begin{figure}[!h]
    \centering
    \includegraphics[width=1\linewidth]{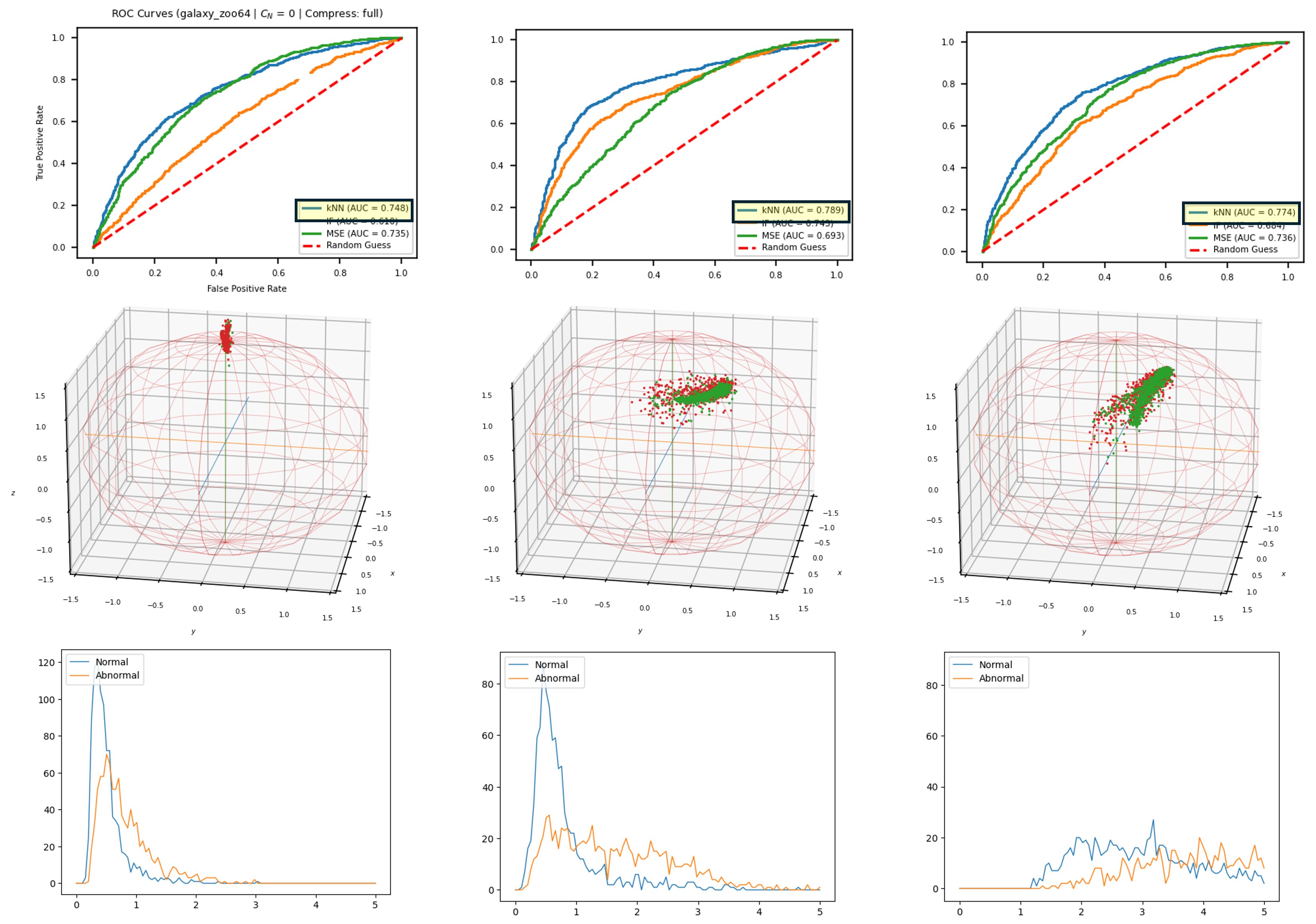}
    \caption{AD metrics improvement at the edge-of-stability in GZ (Middle panel/row: green, normal data; red, anomalies). Same example as that of Fig. \ref{AD_edge_stab}.}
    \label{AD_edge_stabC}
\end{figure}



\end{document}